\documentclass[11pt]{article}

\usepackage[margin=1in]{geometry}
\usepackage{amsmath,amssymb,amsthm,mathtools}
\usepackage{setspace}
\usepackage{natbib}
\usepackage{graphicx}
\usepackage{hyperref}
\usepackage{booktabs}
\usepackage{caption}
\usepackage{enumitem}
\usepackage{tikz}
\usetikzlibrary{arrows.meta,positioning,calc}
\usepackage{pgfplots}
\pgfplotsset{compat=1.18}
\usepgfplotslibrary{fillbetween}

\newtheorem{proposition}{Proposition}
\newtheorem{lemma}{Lemma}
\newtheorem{corollary}{Corollary}

\newtheorem{definition}{Definition}

\newcommand{\E}{\mathbb{E}}
\newcommand{\R}{\mathbb{R}}
\newcommand{\thetabar}{\bar{\theta}}
\newcommand{\ubar}[1]{\underline{#1}}
\newcommand{\dd}{\,\mathrm{d}}

\usepackage{xcolor}

\begin{document}
% ======================================================================

% ---- Title page (unnumbered) ----
\thispagestyle{empty}
\begin{center}
\vspace*{1in}

{\LARGE\bfseries The Screening Cost of Liquidity}

\bigskip\bigskip

{\large Rui Sun\footnote{Haas School of Business, University of California, Berkeley. Email: ruisun233@berkeley.edu.  I am very grateful to Panos Patatoukas, Xiaojun Zhang, Sunil Dutta, Yaniv Konchitchki, and Omri Even-Tov for their guidance and encouragement.  I also thank Federico Echenique, Haluk Ergin, Benjamin Hermalin, Yuichiro Kamada, Shachar Kariv, Giacomo Lanzani, Chris Shannon, and Quitz\'{e} Valenzuela-Stookey for valuable feedback.  All errors are my own.}}

\bigskip

\textit{This version: April 2026}

\vfill

\parbox{5.5in}{\small\textbf{Abstract.}
A principal with cheap capital optimally forces her counterparty to borrow at above-market rates.  The reason: the form of finance is a screening device.  Advances provide liquidity but pool types; contingent transfers separate types but, because they are not pledgeable, impose financing costs.  The optimal contract preserves outside-finance exposure to maintain screening power.  Two sufficient statistics pin down the optimal advance share.  With complementary counterparties, a uniform subsidy that cheapens finance for every relationship can reduce the value of every one.  This explains the coexistence of early payment and contingent compensation in trade credit, venture capital, and internal capital markets.

}
\end{center}

\vfill

\newpage

% ---- Page numbering starts here ----
\setcounter{page}{1}

% ======================================================================
% SECTION 1: INTRODUCTION
% ======================================================================
\section{Introduction}\label{sec:intro}

Whether a firm with access to cheap capital should directly finance its constrained counterparties is a central question in financial contracting.  Large manufacturers routinely contract with smaller suppliers that face tight credit; venture capitalists with deep pockets could write a single check rather than staging funds against milestones; corporate headquarters could fully pre-fund each division's budget rather than tying allocations to performance.  In each case, a natural benchmark suggests that the party with cheaper funds should simply advance the full working-capital requirement, eliminate the counterparty's dependence on expensive outside credit, and capture the resulting surplus.  Yet mixed payment structures that combine early advances with deferred, performance-contingent transfers are ubiquitous.  The benchmark of full pre-financing, like the frictionless ideal it echoes, is logically appealing but empirically counterfactual.\footnote{The parallel with Modigliani and Miller is deliberate: just as the irrelevance result identifies frictions that make capital structure matter, the full-pre-financing benchmark identifies frictions that make the form of inter-firm finance matter.}

In this paper, I study the optimal financial contract between a principal and a counterparty in an environment in which financial contracting serves two distinct functions.  The principal provides liquidity through an advance payment that reduces the counterparty's need to borrow externally at costly terms.  The principal also uses contingent transfers to screen the counterparty's privately known type, separating high-quality counterparties from low-quality ones.

I model the tension between these two roles by assuming that the counterparty's type is private information, that implementing the relationship requires upfront working capital, and that the contingent transfer is not pledgeable to outside lenders.  Non-pledgeability means that only the advance relaxes the borrowing constraint; the contingent leg arrives too late and in too uncertain a form to serve as collateral.  Combined with limited liability, this creates a tight link: more advance means less room for contingency.  The form of finance is therefore a determinant of principal value.  Three main results emerge.

First, the optimal contract equates the marginal liquidity benefit of the advance with the marginal information cost of the resulting contract flattening.  The principal deliberately leaves the counterparty partially exposed to outside finance, even though the principal has cheaper funds.  Full pre-financing would eliminate borrowing costs but would also destroy the information content of the contingent payment, leaving no instrument that discriminates across types.  Retaining outside-finance exposure preserves the contingent leg's ability to tie payments to a verifiable signal, and the resulting type separation is worth the financing cost it entails.

Second, in a two-counterparty extension with complementary relationships, cheaper outside finance for one counterparty can lower principal value.  Complementarity couples the incentive constraints through the joint selection rule: each counterparty's information rent depends on the probability that the other relationship is implemented.  Cheaper finance for one relationship shifts that contract toward more contingency, which expands the set of implemented types for the other counterparty.  The newly included types have high information rents, and the resulting rent increase can exceed the direct financing saving.  The spillover operates through information, not through prices or budgets: the two relationships are financially independent but informationally coupled.  A government subsidy that reduces financing costs for all counterparties simultaneously can reduce the value of every relationship in the firm's portfolio.

Third, the comparative statics have sharp economic content.  Cash intensity rises with outside-credit tightness and falls with signal informativeness.  Tighter financial conditions do not merely raise financing costs; they degrade the quality of screening by forcing the contract toward flatter payment schedules.  Financial conditions affect the real economy through this screening channel, not just through the cost of capital.

The tradeoff requires three ingredients: costly outside finance, non-pledgeability with limited liability, and an informative signal.  Removing any one eliminates the mechanism.  Frictionless outside finance makes the advance payoff-irrelevant.  Pledgeable contingent transfers make the two instruments perfect substitutes.  Unlimited liability lets the principal fully pre-finance and screen through signed repayments.  Uninformative signals make contingent payments useless.

The theory identifies two sufficient statistics for the optimal financial structure: the external-financing cost and the signal informativeness.  A principal who knows these two objects can determine the optimal advance without observing the counterparty's type.  In a parameterized model, the optimal advance ranges from 17\% to 54\% of working capital, bracketing the advance shares in trade-credit contracts \citep{petersen1997trade} and VC staging \citep{kaplan2003financial}.  The screening-contagion effect reduces principal value by up to 15\% at moderate complementarity.

The foundational models of \citet{bolton1990theory} and \citet{holmstrom1997financial} study costly external finance under moral hazard.  The friction here is adverse selection, and the instrument set includes both unconditional and contingent transfers.  The adverse-selection approach to corporate finance originates with \citet{myers1984capital}, who show that information asymmetry between insiders and outsiders distorts financing decisions, and with \citet{stiglitz1981credit}, who show that it can lead to credit rationing.  In the security-design literature, \citet{innes1990limited} characterizes optimal debt contracts under moral hazard with limited liability, \citet{nachman1994optimal} study security design under adverse selection, and \citet{demarzo2005bidding} characterize optimal securities in auction settings.  \citet{demarzo2005optimal} extends the analysis to optimal long-term financing contracts.  I endogenize the mix of an unconditional and a contingent instrument under a screening motive, a different margin from the shape of a single security.  The screening approach connects to the classical analysis of \citet{mussa1978monopoly} and \citet{laffont1986using}, who study optimal pricing and procurement under adverse selection; the novelty here is the interaction between the screening problem and the financing problem through the non-pledgeability constraint.

\citet{burkart2004situ} model trade credit as a response to limited pledgeability: in-kind finance solves a diversion problem.  \citet{cunat2007trade} shows that trade credit serves as insurance against liquidity shocks, and \citet{giannetti2011what} provide evidence that what firms sell is what they lend, linking the form of trade credit to the nature of the traded good.  \citet{klapper2012trade} document the determinants of trade-credit use across countries.  In my framework, the mix of early cash and deferred contingent payments solves a screening problem.  The non-pledgeability friction is shared with \citet{burkart2004situ}, but the economic role differs: theirs prevents diversion; mine preserves information.  The empirical predictions are complementary: both models predict that the form of inter-firm finance responds to the counterparty's credit conditions, but through different channels.

The screening-contagion result connects to the internal capital markets literature.  \citet{stein1997internal} and \citet{scharfstein1998dark} study how constraints propagate through capital reallocation, and \citet{gertner1994internal} compare internal and external capital markets.  \citet{inderst2003laux} study the interaction between headquarters' capital allocation and divisional incentives.  The mechanism here is different: the relationships share no budget but are informationally coupled through complementarity.  The result also relates to the literature on contagion in production networks \citep{acemoglu2012network}, adding an informational propagation channel to the standard price-based mechanisms.

The analytical approach builds on \citet{myerson1981optimal}, adapted to multiple instruments with costly external finance.  \citet{demarzo1999liquidity} show that the form of a security affects its liquidity value; \citet{demarzo2005bidding} show that the form of payment affects screening.  The sufficient-statistics characterization parallels the approach in optimal taxation \citep{chetty2009sufficient} and connects to the broader literature on mechanism design with financial constraints \citep{che1998standard, biais2007optimal}.  \citet{li2021mechanism} studies optimal mechanisms with budget-constrained agents and costly verification, showing that financial constraints fundamentally alter the structure of optimal allocations; the non-pledgeability constraint in this paper plays an analogous role in shaping the optimal contract.  I extend this logic to a setting where the principal chooses a mix of advance and contingent transfer, and the mix itself is the screening device.

The remainder of the paper is organized as follows.  Section~\ref{sec:model} introduces the model.  Section~\ref{sec:bilateral} characterizes the optimal bilateral contract.  Section~\ref{sec:extension} develops the screening-contagion result.  Section~\ref{sec:quant} explores quantitative implications.  Section~\ref{sec:discussion} discusses extensions.  Section~\ref{sec:conclusion} concludes.

\section{Model}\label{sec:model}

The baseline environment of this paper is a standard principal--agent screening model in the tradition of \citet{myerson1981optimal}, extended to allow for costly external finance and multiple payment instruments.  The model combines elements of the financial-contracting literature \citep{bolton1990theory, holmstrom1997financial} with the security-design insight that the form of a financial instrument affects information \citep{demarzo1999liquidity, demarzo2005bidding}.

\medskip
 There are two dates, $t = 0, 1, 2$, and two risk-neutral agents: a principal $P$ and a counterparty $C$.  The counterparty has a privately known type $\theta$, drawn from $[\ubar\theta, \thetabar]$ according to a continuous distribution $F$ with density $f > 0$ on the entire support.  I assume throughout that $F$ satisfies the standard regularity condition: the virtual type $\theta - \frac{1-F(\theta)}{f(\theta)}$ is strictly increasing.\footnote{This regularity condition is satisfied by many common distributions, including the uniform, exponential, and power distributions.  See \citet{myerson1981optimal}.}  Nothing prevents the counterparty from having high or low types; the type captures any dimension of private information that affects the value of the relationship, such as production cost, project quality, or managerial ability.

\medskip
 Implementing the relationship requires date-0 working capital $K > 0$ and generates gross surplus $V(\theta)$ at date~2, with $V' > 0$ and $V'' \leq 0$.  The counterparty bears private cost $c(\theta)$ with $c' > 0$.  Without loss of generality, I normalize so that higher types are more valuable: $V(\theta) - c(\theta)$ is increasing.

The counterparty does not have internal funds and must finance the working capital $K$ from two possible sources.  The principal can provide an advance $a \in [0, K]$ at date~0.  The residual $\ell = K - a$ must be financed from outside lenders at a cost $\Phi(\ell;\, R)$, where $R \geq 0$ indexes the tightness of the counterparty's external credit conditions.  I assume that $\Phi$ satisfies the natural properties: $\Phi(0; R) = 0$ (zero borrowing is costless), $\Phi_\ell > 0$ for $\ell > 0$ (borrowing is costly at the margin), $\Phi_{\ell\ell} \geq 0$ (borrowing cost is convex), and $\Phi_R > 0$ for $\ell > 0$ (tighter markets raise costs).  The quadratic specification $\Phi(\ell; R) = \frac{R}{2}\ell^2$ satisfies all these conditions and is used for the parameterized results.\footnote{The convexity assumption $\Phi_{\ell\ell} \geq 0$ ensures that the principal's problem is well-behaved (the value function is concave in the instruments).  It is satisfied by any financing-cost function that reflects increasing marginal costs of borrowing, which is standard in models of costly external finance.}

\medskip
 After production, a verifiable signal $x \in \mathcal{X}$ is realized.  The signal is drawn from a density $g(x \mid \theta)$ that satisfies the monotone likelihood ratio property (MLRP): higher types are more likely to generate higher signal realizations.  I denote the expected signal by $\mu(\theta) \equiv \E_\theta[x]$; MLRP implies $\mu'(\theta) > 0$.  The signal is informative about type but does not perfectly reveal it: $g(x|\theta) > 0$ for all relevant $(x, \theta)$.  This imperfection matters: if the signal were perfectly revealing, the principal could condition transfers directly on the type, and the screening problem would be trivial.

Two arguments justify the focus on a single signal.  First, in many applications, the verifiable performance measure is naturally one-dimensional (output quantity, quality score, milestone achieved).  Second, as formalized in Proposition~\ref{thm:general}, the main results extend to general monotone securities without requiring the affine specification, so the single-signal structure is without loss of generality for the qualitative results.

\medskip
 The principal offers a contract $(a, T(\cdot))$ at date~0, specifying an advance $a \geq 0$ and a transfer schedule $T: \mathcal{X} \to \R_+$ to be paid at date~2 after observing the signal $x$.  The counterparty accepts or rejects; if it accepts, the principal pays the advance $a$, the counterparty borrows $\ell = K - a$ externally, and production occurs.  At date~2, the signal $x$ is realized and the principal pays $T(x)$.

For the tractable benchmark, I work with the affine transfer class $T(x) = b_0 + b_1 x$ with $b_0, b_1 \geq 0$.  This nests pure advance (full reliance on $a$, with $b_0 = b_1 = 0$), flat contingent transfer ($b_1 = 0$), and contingent-only ($a = 0$, $b_1 > 0$).  The nonnegativity constraint $T(x) \geq 0$ captures limited liability: the counterparty cannot make net repayments to the principal.

\medskip
\noindent\textbf{Assumption.} [NP] \textbf{(Non-pledgeability)}\quad \emph{The contingent transfer $T(x)$ cannot be pledged to outside lenders at date~0.  Only the advance $a$ relaxes the counterparty's date-0 borrowing requirement.}

\medskip

Assumption [NP] is the key restriction that distinguishes this model from a standard screening problem.  Combined with limited liability, it prevents the principal from fully pre-financing the counterparty and then screening through a state-contingent repayment schedule.  Without [NP], the counterparty could borrow against the present value of $\E_\theta[T(x)]$, making the advance and the contingent transfer perfect substitutes in the financing dimension.  The tradeoff between liquidity and screening would disappear.

Non-pledgeability is natural in three settings.  First, if the signal is relationship-specific, observable only within the principal--counterparty pair, then outside lenders cannot contract on it and therefore cannot price $T(x)$.  Second, if the contingent payment is renegotiable, lenders discount its value as collateral because the principal and counterparty may agree to reduce it ex post.  Third, in many trade-credit and supplier-finance markets, assignment of buyer-specific receivables is legally possible but practically limited by consent requirements and offset rights.  Assumption [NP] captures this imperfect assignability in reduced form.

I explicitly state when Assumption [NP] is used in the paper.  All results that do not invoke [NP] hold in the standard screening environment with pledgeable transfers.

\medskip
 A type-$\theta$ counterparty who accepts the contract $(a, b_0, b_1)$ receives expected payoff
\begin{equation}\label{eq:U}
U(\theta) \;=\; a + b_0 + b_1\,\mu(\theta) - c(\theta) - \Phi(K - a;\, R).
\end{equation}
The principal's payoff from a type-$\theta$ counterparty is
\begin{equation}\label{eq:Pi}
\Pi_P(\theta) \;=\; V(\theta) - a - b_0 - b_1\,\mu(\theta).
\end{equation}
The counterparty participates if and only if $U(\theta) \geq 0$.

The economic content of Assumption [NP] is visible directly in equation~\eqref{eq:U}.  The advance $a$ appears in both the transfer ($+a$) and the financing cost ($-\Phi(K-a; R)$): it simultaneously compensates and relaxes borrowing.  The contingent component $b_0 + b_1\mu(\theta)$ appears only in the transfer: it compensates but does not relax borrowing.  This asymmetry is what gives the two legs different economic roles and is the source of the liquidity--screening tradeoff.

Figure~\ref{fig:payoff} illustrates the counterparty's expected payoff $U(\theta)$ under three contract structures, holding total expected payment constant.  Under a pure advance contract ($a = K$, $b_1 = 0$), the payoff is flat in the type: there is no screening, and all types receive the same net transfer.  Under a pure contingent contract ($a = 0$, $b_1 > 0$), the payoff is steeply increasing; high types earn substantially more, providing strong separation, but the financing cost $\Phi(K; R)$ is maximal because no advance is provided.  The optimal contract mixes the two instruments, producing an intermediate slope that balances screening against financing cost.

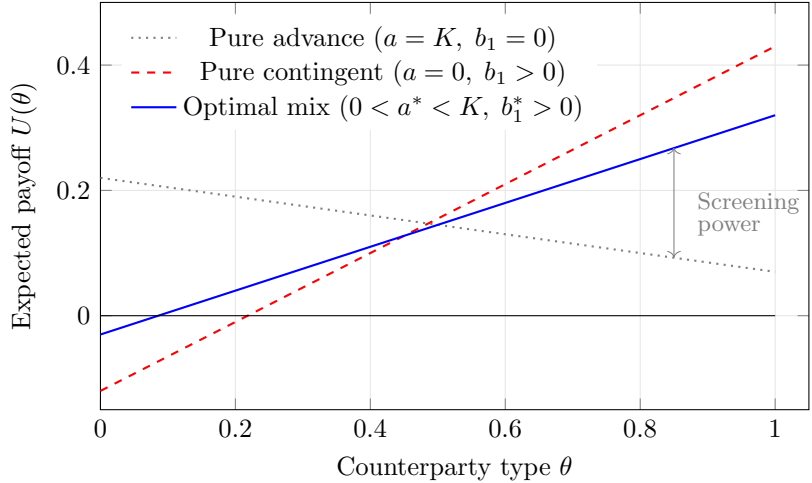
\begin{figure}[t]
\centering
\begin{tikzpicture}
\begin{axis}[
    width=11cm, height=7cm,
    xlabel={Counterparty type $\theta$},
    ylabel={Expected payoff $U(\theta)$},
    xmin=0, xmax=1.05,
    ymin=-0.15, ymax=0.5,
    ymajorgrids=true,
    xmajorgrids=true,
    grid style={gray!20},
    every axis plot/.append style={thick},
    tick label style={font=\small},
    label style={font=\small},
    legend pos=north west,
    legend style={font=\small, draw=none},
]

% Pure advance: decreasing payoff (no screening, U = K - c(theta))
\addplot[gray, dotted, domain=0:1, samples=50]
    {0.22 - 0.15*x};
\addlegendentry{Pure advance ($a = K,\; b_1 = 0$)}

% Pure contingent: steep but high financing cost
\addplot[red, dashed, domain=0:1, samples=50]
    {-0.12 + 0.55*x};
\addlegendentry{Pure contingent ($a = 0,\; b_1 > 0$)}

% Optimal mix: intermediate slope
\addplot[blue, domain=0:1, samples=50]
    {-0.03 + 0.35*x};
\addlegendentry{Optimal mix ($0 < a^* < K,\; b_1^* > 0$)}

% Zero line
\addplot[black, thin, domain=0:1] {0};

% Annotations
\draw[<->, gray, thin] (axis cs:0.85,0.093) -- (axis cs:0.85,0.268);
\node[font=\footnotesize, gray, anchor=west] at (axis cs:0.87,0.18) {Screening};
\node[font=\footnotesize, gray, anchor=west] at (axis cs:0.87,0.14) {power};

\end{axis}
\end{tikzpicture}
\caption{\textbf{Counterparty payoff under alternative contract structures.}  The figure illustrates $U(\theta)$ for three contracts.  Under the pure advance (dotted), the principal fully pre-finances the counterparty ($a = K$, $b_1 = 0$); the payoff $U(\theta) = K - c(\theta)$ is decreasing because higher types have higher costs, but the contract cannot distinguish types.  Under the pure contingent contract (dashed), the payoff is steeply increasing because high types earn large signal-contingent payments, but the financing cost $\Phi(K; R)$ is maximal.  The optimal contract (solid) balances the two: the slope $b_1^*\mu'(\theta) - c'(\theta)$ is intermediate, providing screening power while partially relieving the financing burden.}
\label{fig:payoff}
\end{figure}

\medskip
 By the revelation principle, I work with a direct mechanism in which the counterparty reports $\hat\theta$ and the principal implements an allocation rule $q(\hat\theta) \in \{0,1\}$ and contract terms $(a(\hat\theta), b_0(\hat\theta), b_1(\hat\theta))$.

\begin{definition}[Implementable contract]
A contract $(q, a, b_0, b_1)$ is implementable if it satisfies incentive compatibility and individual rationality.
\end{definition}

The principal maximizes expected profit over the set of implementable contracts:
\begin{equation}\label{eq:program}
\max_{\substack{q(\cdot),\, a(\cdot),\\ b_0(\cdot),\, b_1(\cdot)}} \;\; \int_{\ubar\theta}^{\thetabar} q(\theta)\,\Pi_P(\theta)\; f(\theta)\dd\theta
\end{equation}
subject to incentive compatibility, individual rationality, and limited liability ($b_0 \geq 0$, $b_1 \geq 0$, $a \in [0,K]$).  The characterization of this problem, including the derivation of implementability conditions, the virtual surplus, and the optimal contract, is the subject of Section~\ref{sec:bilateral}.

\medskip
\noindent Two-counterparty extension.\label{sec:model_ext}\quad For the results in Section~\ref{sec:extension}, I extend the model to two counterparties $i \in \{1, 2\}$ with independently drawn types $\theta_i \sim F_i$ on $[\ubar\theta_i, \thetabar_i]$.  Each requires working capital $K_i$, faces external-financing cost $\Phi_i(\ell_i; R_i)$, and receives contract $(a_i, b_{0,i}, b_{1,i})$.  The principal's payoff is
\begin{equation}\label{eq:Pi_ext}
\Pi_P = \sum_{i=1}^{2} q_i\bigl[V_i(\theta_i) - a_i - b_{0,i} - b_{1,i}\mu_i(\theta_i)\bigr] + \delta\, q_1 q_2, \qquad \delta > 0.
\end{equation}
The term $\delta\, q_1 q_2$ captures complementarity: the principal's payoff from implementing both relationships exceeds the sum of implementing each alone.  This arises when one counterparty's output is an input to the other, or when the principal's downstream technology has super-modular returns.  The contracts are designed jointly but offered separately.  The coupling through the joint selection rule is what generates the screening spillover in Section~\ref{sec:extension}.

% ======================================================================
% SECTION 3: OPTIMAL BILATERAL CONTRACT
% ======================================================================
\section{Optimal Bilateral Contract}\label{sec:bilateral}

All proofs are in the Appendix.

\medskip
 Consider a type-$\theta$ counterparty who reports $\hat\theta$ under the direct mechanism.  Its payoff from misreporting is
\begin{equation}\label{eq:Umis}
U(\hat\theta, \theta) \;=\; a(\hat\theta) + b_0(\hat\theta) + b_1(\hat\theta)\,\mu(\theta) - c(\theta) - \Phi(K - a(\hat\theta);\, R).
\end{equation}
The counterparty's true type $\theta$ enters equation~\eqref{eq:Umis} only through $\mu(\theta)$ and $c(\theta)$.  All other terms depend on the reported type $\hat\theta$.  This separability is what makes the problem tractable: it means that the counterparty's incentive to misreport depends only on how the contingent transfer $b_1$ interacts with the type through the signal, not on how the advance $a$ interacts with the type through financing.

Truthful reporting requires $\hat\theta = \theta$ to be optimal for all types.  At the truthful report, the standard envelope argument \citep{milgrom2002envelope} gives the local incentive-compatibility condition:
\begin{equation}\label{eq:envelope}
\frac{dU}{d\theta} \;=\; b_1(\theta)\,\mu'(\theta) \;-\; c'(\theta).
\end{equation}
Equation~\eqref{eq:envelope} reveals the central separation in the model.  The advance $a$ does not appear in the local incentive constraint: it affects the level of the counterparty's payoff (through the financing cost $\Phi$) but not its slope with respect to type.  Only the contingent slope $b_1$ governs how fast the counterparty's payoff rises with type.  Economically, the advance is pure liquidity, helping all types equally, while the contingent leg is the screening instrument.  This separation is the source of the tradeoff: the principal cannot simultaneously maximize liquidity and maximize screening, because the two objectives load on different instruments that compete for the same budget.

Global incentive compatibility further requires that the contingent slope $b_1(\theta)$ is non-decreasing in the type.  Under the regularity condition on $F$, this is equivalent to requiring that the virtual surplus is non-decreasing in $\theta$.  The formal necessity and sufficiency of these conditions is established in the Appendix using the standard two-inequality technique \citep{myerson1981optimal}.

\medskip
 A simplifying result reduces the dimension of the problem.

\begin{lemma}[Optimality of zero intercept]\label{lem:b0zero}
At the optimum, $b_0^* = 0$.  The entire unconditional component of the contract is delivered as the advance $a^*$, and the date-2 transfer loads entirely on the signal: $T^*(x) = b_1^*\, x$.
\end{lemma}

Intuitively, unconditional cash is always better delivered early.  A dollar of $b_0$ arriving at date~2 does the same compensating job as a dollar of $a$ arriving at date~0 but misses the opportunity to reduce external borrowing.  Because $b_0$ does not enter the envelope condition~\eqref{eq:envelope}, replacing $b_0$ with an equal increase in $a$ leaves the information rent unchanged while reducing the financing cost by $\Phi_\ell \cdot b_0 > 0$.  Any contract with $b_0 > 0$ is therefore strictly dominated.  This dimensional reduction is standard in mechanism-design problems with multiple instruments.

Lemma~\ref{lem:b0zero} reduces the instrument space from three variables $(a, b_0, b_1)$ to two $(a, b_1)$.  The binding participation constraint at the lowest type provides a further reduction.  With $b_0 = 0$ and $U(\ubar\theta) = 0$ (which is optimal under regularity), the participation constraint becomes
\begin{equation}\label{eq:binding_IR}
a + b_1\,\mu(\ubar\theta) \;=\; c(\ubar\theta) + \Phi(K - a;\, R).
\end{equation}
This equation implicitly defines $a$ as a function of $b_1$:

\begin{lemma}[Instrument-space reduction]\label{lem:reduction}
The binding participation constraint~\eqref{eq:binding_IR} defines a smooth function $a = a(b_1;\, R)$ with
\begin{equation}\label{eq:da_db1}
\frac{da}{db_1} \;=\; \frac{-\mu(\ubar\theta)}{1 + \Phi_\ell(K - a;\, R)} \;<\; 0.
\end{equation}
More contingency requires less advance.
\end{lemma}

The derivative is negative because increasing the contingent slope $b_1$ raises the expected transfer to the lowest type, which must be offset by reducing the advance to maintain the binding participation constraint.  The magnitude depends on the marginal financing cost $\Phi_\ell$: when outside finance is expensive, a given increase in $b_1$ requires a smaller reduction in $a$, because the financing-cost saving from the advance is large.

With this reduction, the principal's problem becomes a single-variable optimization over $b_1 \geq 0$, with $a(b_1)$ adjusting endogenously:
\begin{equation}\label{eq:reduced_problem}
\max_{b_1 \geq 0}\;\; W(b_1) \;\equiv\; \int_{\ubar\theta}^{\thetabar} q(\theta)\;\psi\bigl(\theta;\, a(b_1),\, b_1\bigr)\; f(\theta)\dd\theta \;-\; a(b_1),
\end{equation}
where $\psi$ is the virtual surplus, derived next.

\medskip
 The principal's expected profit, after substituting the counterparty's payoff and using $U(\ubar\theta) = 0$, can be decomposed using the standard integration-by-parts technique \citep{myerson1981optimal}.  Expected information rents are
\begin{equation}\label{eq:expected_rents}
\E\bigl[U(\theta)\,q(\theta)\bigr] \;=\; \int_{\ubar\theta}^{\thetabar} q(\theta)\;b_1\,\mu'(\theta)\,\frac{1 - F(\theta)}{f(\theta)}\; f(\theta)\dd\theta.
\end{equation}
The hazard rate $\frac{1-F(\theta)}{f(\theta)}$ captures the cost of incentive compatibility: each implemented type $\theta$ must be compensated not only for its own cost but also for the temptation that lower types face to mimic it.  The product $b_1\mu'(\theta)$ multiplying the hazard rate is the screening intensity, measuring how fast the contract separates types through the signal.

Substituting~\eqref{eq:expected_rents} into the principal's expected payoff and collecting terms, the principal's objective can be written as equation~\eqref{eq:reduced_problem}, where the virtual surplus is
\begin{equation}\label{eq:VS_main}
\psi(\theta;\, a, b_1) \;=\; V(\theta) - c(\theta) - \Phi(K - a;\, R) - b_1\,\mu'(\theta)\,\frac{1 - F(\theta)}{f(\theta)}.
\end{equation}
The optimal allocation rule is $q^*(\theta) = \mathbf{1}\{\psi(\theta) \geq 0\}$: the principal implements the relationship whenever the virtual surplus is non-negative.

Equation~\eqref{eq:VS_main} is the paper's core object.  It has the same architecture as the standard virtual surplus in mechanism design \citep{myerson1981optimal}: productive value minus direct cost minus a hazard-rate information-rent distortion.  The novelty is the financing wedge $\Phi(K - a; R)$, which is endogenous in the screening instrument.  Because $a$ and $b_1$ are linked through the participation constraint~\eqref{eq:binding_IR}, choosing more contingency ($b_1$) forces less advance ($a$), which increases the financing wedge.  It is this coupling between the financing wedge and the information rent that generates the liquidity--screening tradeoff.  In a standard screening model without financial frictions, the financing wedge is absent, and the optimal contract is determined by the information rent alone.

To make this coupling explicit, it is useful to decompose the principal's expected value into four components:
\begin{equation}\label{eq:decomposition}
W(b_1) \;=\; \underbrace{S\bigl(\hat\theta(b_1)\bigr)}_{\text{productive surplus}} \;-\; \underbrace{\Phi\bigl(K - a(b_1);\, R\bigr)\cdot\bigl[1 - F(\hat\theta)\bigr]}_{\text{aggregate financing cost}} \;-\; \underbrace{\int_{\hat\theta}^{\thetabar} b_1\mu'(\theta)\frac{1-F(\theta)}{f(\theta)}f(\theta)d\theta}_{\text{aggregate information rent}} \;-\; \underbrace{a(b_1)}_{\text{advance outlay}},
\end{equation}
where $S(\hat\theta) = \int_{\hat\theta}^{\thetabar}[V(\theta) - c(\theta)]f(\theta)d\theta$ is the productive surplus from implemented types and $\hat\theta(b_1)$ is the implementation cutoff.  Equation~\eqref{eq:decomposition} reveals that a marginal increase in $b_1$ has three effects on principal value, operating through three distinct channels.  First, it changes the set of implemented types (through $\hat\theta$), affecting the productive surplus.  Second, it reduces the advance (through the binding IR~\eqref{eq:da_db1}), which raises the financing cost per implemented type.  Third, it directly increases the aggregate information rent.  The first-order condition~\eqref{eq:FOC_main} balances all three margins simultaneously.

\medskip
\noindent\textbf{The optimal contract}\quad The main result of the paper characterizes the optimal bilateral contract.  Proposition~\ref{thm:bilateral} shows that the optimal contract uses both instruments at interior levels and provides the first-order condition that determines the optimum.

\begin{proposition}[Optimal preservation of outside-finance exposure]\label{thm:bilateral}
Suppose the regularity conditions and Assumption~[NP] hold and $R > 0$.  Then:
\begin{enumerate}[label=\alph*)]
\item The optimal advance is interior: $0 < a^* < K$.  The counterparty retains outside-finance exposure $\ell^* = K - a^* > 0$ despite the principal being the cheaper source of funds.
\item The optimal contingent schedule is strictly increasing: $b_1^* > 0$.
\item The optimum is characterized by the first-order condition
\begin{equation}\label{eq:FOC_main}
\underbrace{\Phi_\ell(K - a^*;\, R)}_{\substack{\text{marginal liquidity}\\\text{benefit}}} \;=\; \underbrace{\frac{\partial}{\partial a}\;\E\!\left[b_1^*\,\mu'(\theta)\,\frac{1-F(\theta)}{f(\theta)}\;\bigg|\; q = 1\right]}_{\substack{\text{marginal information}\\\text{cost}}}.
\end{equation}
\end{enumerate}
\end{proposition}

 Proposition~\ref{thm:bilateral}a) shows that the principal optimally sets the advance strictly between zero and the full working-capital requirement.  This result may appear counterintuitive: the principal has cheaper funds than the counterparty, so standard logic suggests she should eliminate all borrowing costs by setting $a = K$.  The reason she does not is that full pre-financing would require a completely flat contingent schedule ($b_1 = 0$), which destroys screening.  At $a = K$, the binding participation constraint~\eqref{eq:binding_IR} leaves no room for $b_1 > 0$, so the contract cannot distinguish across types.  By retaining some outside-finance exposure, the principal preserves the ability to tie payments to the signal and separate types.  The advance is not a financing decision; it is an information decision.

\medskip
 Proposition~\ref{thm:bilateral}b) establishes that the contingent schedule is strictly increasing in the signal.  Combined with part~a), this means the optimal contract is genuinely mixed: neither pure advance nor pure contingency, but a combination that balances liquidity and screening.  The contingent slope $b_1^*$ is the principal's choice of how much screening power to embed in the payment structure.  Higher $b_1^*$ means more separation across types but also more financing exposure for the counterparty.  In the other direction, at $a = 0$, the counterparty bears the full financing cost $\Phi(K; R)$, and the gain from introducing a marginal advance, reducing this cost at rate $\Phi_\ell(K; R) > 0$, exceeds the loss from slightly reduced screening.  The optimum is therefore interior on both margins.

\medskip
 Proposition~\ref{thm:bilateral}c) characterizes the interior through a condition that equates two marginal objects.  The left side of equation~\eqref{eq:FOC_main} is the marginal liquidity benefit: the reduction in external-financing cost from one additional dollar of advance.  The right side is the marginal information cost: the increase in expected information rents that results from the contract becoming flatter as the advance rises (through the binding participation constraint).  This is not a generic interior optimum: it has a specific economic interpretation as the principal choosing the information sensitivity of the financial contract.  A principal facing high $R$ (expensive outside finance) optimally tilts toward more advance and accepts lower screening power.  A principal with access to a highly informative signal optimally tilts toward more contingency and accepts higher financing costs.  The first-order condition~\eqref{eq:FOC_main} formalizes exactly where this tilting stops.

The optimal transfer schedule has a natural decomposition.  Define the first-best contingent slope $b_1^{FB}$ as the value that would be optimal absent information frictions \footnote{That is, if the principal could observe $\theta$ directly.}.  Under complete information, the principal can extract all surplus, so $b_1^{FB}$ is determined solely by the financing condition: it is the slope that minimizes the counterparty's financing cost for the given surplus extraction.  Under asymmetric information, the optimal slope $b_1^*$ differs from $b_1^{FB}$ by an information adjustment:
\begin{equation}\label{eq:transfer_decomp}
b_1^* \;=\; b_1^{FB} \;+\; \underbrace{\Delta b_1(R, \mu')}_{\text{information adjustment}},
\end{equation}
where $\Delta b_1 > 0$ when information rents are decreasing in $b_1$ (the principal increases contingency to improve screening) and $\Delta b_1 < 0$ when financing costs are high (the principal decreases contingency to save on financing).  The sign and magnitude of the information adjustment depend on $R$ and $\mu'$, and the first-order condition~\eqref{eq:FOC_main} pins down $\Delta b_1$ exactly.

This decomposition implies that the optimal contract can be viewed as a first-best financial structure, corrected for the information friction.  The correction is small when $R$ is low (information rents are cheap relative to financing savings) and large when $R$ is high (financing costs dominate, pushing the contract toward the advance-heavy first-best).

\begin{proposition}[Sufficient statistics]\label{prop:sufficient}
Under the conditions of Proposition~\ref{thm:bilateral}, the optimal contract $(a^*, b_1^*)$ depends on the primitives $(F, V, c, \Phi, \mu)$ only through two objects:
\begin{enumerate}[label=\alph*)]
\item the marginal financing cost at the optimum, $\Phi_\ell(K - a^*;\, R)$, which is a sufficient statistic for the financing environment;
\item the marginal screening return at the optimum, $\frac{\partial}{\partial b_1}\E[b_1\mu'(\theta)\frac{1-F}{f} \mid q=1]$, which is a sufficient statistic for the information environment.
\end{enumerate}
A principal who knows these two objects can determine the optimal advance share without observing the counterparty's type distribution directly.
\end{proposition}

Proposition~\ref{prop:sufficient} is the analog of the sufficient-statistics characterization in optimal tax theory: just as knowledge of a demand elasticity and a welfare weight is sufficient to determine the optimal tax rate, knowledge of the financing semi-elasticity and the screening semi-elasticity is sufficient to determine the optimal financial contract.  The principal does not need to estimate the full type distribution $F$ or the full signal technology $g(x|\theta)$; she needs only the local curvatures of the financing cost function and the information rent function at the relevant margin.  In the parameterized model, these reduce to two scalar parameters: $R$ (which determines $\Phi_\ell$) and $\mu'$ (which determines the screening return per unit of contingency).

The sufficient-statistics result yields a one-line implementation rule: advance funds to the counterparty until the counterparty's marginal borrowing rate equals the principal's marginal screening return per dollar of contingency.  A principal who observes that her counterparty's marginal external-financing cost is 8\% and that the marginal screening return from an additional dollar of contingent payment is 8\% should stop advancing.  If the borrowing rate exceeds the screening return, the contract is too contingent; if it falls below, the contract is too advance-heavy.  This rule can be applied without knowing the counterparty's type distribution, just as Dávila's volume-targeting rule can be applied without knowing investors' beliefs.

\medskip
The bilateral theorem characterizes the optimal contract and the forces that determine it.  A natural next question is: how does the contract respond to changes in the economic environment?  Two primitives are directly observable in practice: the tightness of outside credit markets and the informativeness of available performance measures.  Proposition~\ref{thm:bilateral} implies sharp comparative statics with respect to both.

\begin{corollary}\label{cor:compstat}
Under the conditions of Proposition~\ref{thm:bilateral}:
\begin{enumerate}[label=\alph*)]
\item Cash intensity rises with outside-finance tightness: $\partial\beta^*/\partial R > 0$.
\item Cash intensity falls with signal informativeness: if the signal becomes more informative in the Blackwell sense, $b_1^*$ increases and $\beta^*$ decreases.
\item Tighter financial conditions worsen selection: higher $R$ forces the contract toward flatter payment schedules, widening the gap between the second-best allocation and the first-best.
\end{enumerate}
\end{corollary}

Corollary~\ref{cor:compstat}a) predicts that counterparties in tighter credit environments receive more early payment and less performance-contingent compensation.  Intuitively, when outside borrowing is expensive, the principal substitutes toward the liquidity instrument (advance) and away from the screening instrument (contingency).  The mechanism operates through the first-order condition~\eqref{eq:FOC_main}: higher $R$ raises the marginal liquidity benefit (the left side), shifting the optimum toward more advance.  In supply-chain finance, this predicts that suppliers facing credit crunches receive more pre-shipment funding and fewer deferred quality-based payments.

\medskip
Corollary~\ref{cor:compstat}b) predicts that relationships with better verifiable performance measures rely more heavily on contingent transfers.  When the signal is highly informative, each unit of contingency buys more separation across types at a lower information-rent cost per unit.  The principal shifts toward more screening and less liquidity.  In venture capital, this predicts that startups with easily measurable milestones \footnote{For instance, FDA approval or revenue targets.} face more aggressive staging than startups whose quality is harder to verify.

\medskip
Corollary~\ref{cor:compstat}c) establishes that tighter financial conditions do not merely raise financing costs; they also degrade the quality of screening.  Higher $R$ forces the contract toward flatter payment schedules, which reduces the principal's ability to distinguish across types.  Fewer types are implemented, and the gap between the second-best and first-best allocations widens.  This link between financial conditions and real screening distortions is a distinctive prediction of the model: credit crunches have informational consequences, not just financing consequences.

\medskip
The optimal contract uses both instruments at interior levels, but a practitioner might reasonably ask: how much is lost by using a simpler structure?  A large buyer could advance the full working capital and avoid the complexity of contingent payments.  Alternatively, a venture capitalist could stage everything against milestones and provide no upfront funding.  The following result shows that both shortcuts are strictly dominated, and quantifies the welfare gap.

Define three value functions.  Let $W_A(R)$ denote the principal's value under a pure advance contract ($a = K$, $b_1 = 0$): the principal fully pre-finances the counterparty and imposes no contingent payment.  Let $W_C(R)$ denote the principal's value under a pure contingent contract ($a = 0$, $b_1$ chosen optimally): the principal provides no advance and relies entirely on the contingent transfer for both compensation and screening.  Let $W_M(R)$ denote the principal's value under the optimal mixed contract ($a^*$, $b_1^*$ chosen jointly).

\begin{proposition}[Strict dominance of the mixed contract]\label{prop:dominance}
Suppose the regularity conditions and Assumption~[NP] hold.  Then:
\begin{enumerate}[label=\alph*)]
\item For all $R > 0$ with an informative signal, $W_M(R) > \max\{W_A(R),\, W_C(R)\}$: the mixed contract strictly dominates both pure instruments.
\item $W_A(R)$ is independent of $R$: under full pre-financing, the counterparty bears no financing cost, so credit conditions are irrelevant.
\item $W_C(R)$ is strictly decreasing in $R$: under pure contingency, the counterparty bears the full financing cost $\Phi(K;\, R)$, which rises with credit tightness.
\item There exists a unique threshold $R^*$ such that $W_C(R) > W_A(R)$ for $R < R^*$ and $W_A(R) > W_C(R)$ for $R > R^*$.  At $R = 0$, $W_C(0) > W_A(0)$: screening is valuable when finance is frictionless.
\end{enumerate}
\end{proposition}

Proposition~\ref{prop:dominance}a) is the central welfare result.  It says that neither pure instrument can replicate what the mixed contract achieves: the advance provides liquidity that the contingent transfer cannot, and the contingent transfer provides screening that the advance cannot.  By combining both, the principal exploits the comparative advantage of each instrument.  This strict dominance holds for all positive financing frictions, not just for a knife-edge set of parameters.

Proposition~\ref{prop:dominance}b) and c) reveal the asymmetry between the two instruments.  The pure advance insulates the counterparty from credit conditions entirely, but at the cost of forgoing all screening.  The pure contingent contract preserves screening, but exposes the counterparty to the full financing cost.  As $R$ rises, the pure contingent contract becomes progressively more expensive, while the pure advance is unaffected.  Eventually, the financing cost under pure contingency exceeds the value of screening, and the pure advance dominates.

The threshold $R^*$ in Proposition~\ref{prop:dominance}d) has a natural economic interpretation: it is the level of financing tightness at which the screening benefit of contingent payments is exactly offset by the financing cost they impose.  Below $R^*$, credit is cheap enough that the counterparty can afford the financing exposure, and the screening benefit is worth it.  Above $R^*$, credit is so expensive that the financing cost outweighs the screening gain, and pooling through a pure advance becomes preferable.  The mixed contract avoids this corner by trading off liquidity and screening at the margin, which is always strictly better than either extreme.

In the uniform-quadratic parameterization, $W_A$ can be computed in closed form: with $a = K = 1$ and $b_1 = 0$, the financing cost is zero, and the principal implements all types $\theta$ with $V(\theta) > K$, giving $W_A = \int_{1/v}^{1}(v\theta - 1)d\theta$.  The threshold $R^*$ can also be computed explicitly.  Figure~\ref{fig:dominance} illustrates the three value functions.

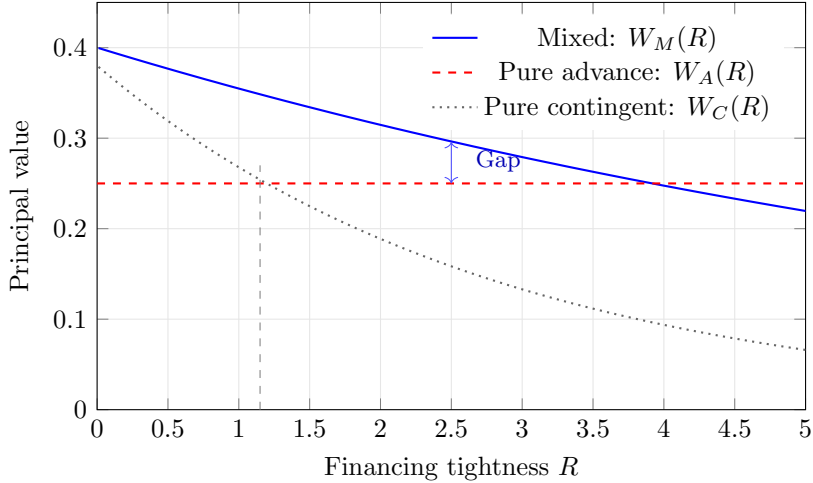
\begin{figure}[t]
\centering
\begin{tikzpicture}
\begin{axis}[
    width=11cm, height=7cm,
    xlabel={Financing tightness $R$},
    ylabel={Principal value},
    xmin=0, xmax=5,
    ymin=0, ymax=0.45,
    ymajorgrids=true,
    xmajorgrids=true,
    grid style={gray!20},
    every axis plot/.append style={thick},
    tick label style={font=\small},
    label style={font=\small},
    legend pos=north east,
    legend style={font=\small, draw=none},
]
% Mixed (highest, decreasing slowly)
\addplot[blue, domain=0.01:5, samples=80]
    {0.40*exp(-0.12*x)};
\addlegendentry{Mixed: $W_M(R)$}

% Pure advance (constant, no R dependence)
\addplot[red, dashed, domain=0:5]
    {0.25};
\addlegendentry{Pure advance: $W_A(R)$}

% Pure contingent (starts high, decreases fast)
\addplot[black!60, dotted, domain=0.01:5, samples=80]
    {0.38*exp(-0.35*x)};
\addlegendentry{Pure contingent: $W_C(R)$}

% R* crossing point
\draw[gray, thin, dashed] (axis cs:1.15,0) -- (axis cs:1.15,0.27);
\node[font=\footnotesize, gray, anchor=north] at (axis cs:1.15,0) {$R^*$};

% Annotations
% Dominance gap arrow at R=2.5
\draw[<->, blue!70, thin] (axis cs:2.5,0.25) -- (axis cs:2.5,0.296);
\node[font=\footnotesize, blue!70!black, anchor=west] at (axis cs:2.6,0.275) {Gap};

\end{axis}
\end{tikzpicture}
\caption{\textbf{Strict dominance of the mixed contract.}  The mixed contract $W_M(R)$ (solid) strictly dominates both pure instruments for all $R > 0$.  The pure advance $W_A$ (dashed) is independent of $R$: full pre-financing eliminates the financing friction but forfeits screening.  The pure contingent contract $W_C$ (dotted) starts above $W_A$ when credit is cheap but falls below it at $R = R^*$: the screening benefit no longer justifies the financing cost.  The mixed contract avoids both corners by optimally trading off liquidity and screening.}
\label{fig:dominance}
\end{figure}

The welfare ordering in Proposition~\ref{prop:dominance} provides a precise answer to the question posed in the introduction: why does a principal with cheap capital not simply pre-finance the constrained counterparty?  The answer is that full pre-financing ($W_A$) is dominated by the mixed contract ($W_M$) because it sacrifices screening.  But the answer is nuanced: pure contingency ($W_C$) is also dominated, because it imposes excessive financing costs.  The optimal contract lies strictly between the two extremes, and the principal's value exceeds what either pure instrument can deliver.  This strict dominance result is the theoretical foundation for the quantitative analysis in Section~\ref{sec:quant}, where the magnitudes of the dominance gap are computed explicitly.

The source of the strict dominance can be traced to the virtual surplus, which takes a different form under each contract structure.  Under a pure advance ($a = K,\; b_1 = 0$), the virtual surplus is
\begin{equation}\label{eq:VS_advance}
\psi_A(\theta) \;=\; V(\theta) - c(\theta) - K.
\end{equation}
The financing wedge vanishes ($\Phi(0; R) = 0$) and so does the information-rent distortion.  The principal implements all types above $V(\theta) \geq K$ and cannot separate among them: good and bad types are pooled.

Under a pure contingent contract ($a = 0,\; b_1$ chosen optimally), the virtual surplus is
\begin{equation}\label{eq:VS_contingent}
\psi_C(\theta;\, b_1) \;=\; V(\theta) - c(\theta) - \Phi(K;\, R) - b_1\mu'(\theta)\frac{1-F(\theta)}{f(\theta)}.
\end{equation}
Screening is maximized: the hazard-rate term is as steep as possible.  But the financing wedge is also maximal: $\Phi(K; R)$ is the cost of financing the entire working-capital requirement externally.  Every implemented type bears this full cost.

Under the mixed contract, the virtual surplus~\eqref{eq:VS_main} replaces $\Phi(K; R)$ with $\Phi(K - a^*; R) < \Phi(K; R)$ while retaining a positive screening term $b_1^* > 0$.  The mixed contract achieves strictly higher value by exploiting the comparative advantage of each instrument: the advance handles the financing dimension at lower cost (because it directly reduces borrowing), while the contingent leg handles the information dimension at lower cost (because it is more informative per dollar than the advance).

The information rent has a correspondingly different structure under each mechanism.  Under the pure advance, every type earns the same surplus; the rent schedule is flat.
\begin{equation}\label{eq:rent_advance}
U_A(\theta) \;=\; K - c(\theta), \qquad \text{for all implemented } \theta.
\end{equation}
There is no type separation: the principal cannot price-discriminate based on quality.  Under the pure contingent contract, the rent schedule has maximum slope:
\begin{equation}\label{eq:rent_contingent}
U_C(\theta) \;=\; \int_{\ubar\theta}^{\theta}\bigl[b_1^C\,\mu'(\theta') - c'(\theta')\bigr]d\theta',
\end{equation}
where $b_1^C$ is the optimal slope under pure contingency.  High types earn large rents; low types earn small rents.  The slope is $dU_C/d\theta = b_1^C\mu'(\theta) - c'(\theta)$, which is the source of screening power.  Under the mixed contract, the rent schedule lies between these extremes:
\begin{equation}\label{eq:rent_mixed}
U_M(\theta) \;=\; \int_{\ubar\theta}^{\theta}\bigl[b_1^*\,\mu'(\theta') - c'(\theta')\bigr]d\theta', \qquad b_1^* \in (0,\, b_1^C).
\end{equation}
The slope is $b_1^*\mu' - c' < b_1^C\mu' - c'$: less separation than pure contingency, but the saving in financing costs more than compensates.  The binding participation constraint~\eqref{eq:binding_IR} links the two: since $a^* > 0$ requires $b_1^* < b_1^C$, the principal trades off rent steepness (screening) for advance size (liquidity) at the rate determined by the first-order condition~\eqref{eq:FOC_main}.

The three virtual surpluses~\eqref{eq:VS_advance},~\eqref{eq:VS_contingent}, and~\eqref{eq:VS_main} encapsulate the liquidity--screening tradeoff in a single comparison.  The advance removes the financing wedge but eliminates the screening term.  The contingent contract maximizes the screening term but also maximizes the financing wedge.  The mixed contract optimally positions itself in the interior, where the marginal financing saving from a dollar of advance equals the marginal information cost from the resulting contract flattening.

\medskip
The results so far use the affine transfer class $T(x) = b_0 + b_1 x$.  A reader might worry that the liquidity-screening tradeoff is an artifact of this particular functional form.  The next result shows it is not: the tradeoff arises for any monotone payment schedule, including equity-like, debt-like, or option-like securities.

\begin{proposition}[General security class]\label{thm:general}
Let $T(x)$ belong to the class $\mathcal{T}$ of nonnegative monotone increasing functions on $\mathcal{X}$.  Order this class by steepness: $T_s(x)$ is steeper than $T_{s'}(x)$ for $s > s'$ if $\E_\theta[T_s(x)] - \E_{\theta'}[T_s(x)]$ is increasing in $s$ for all $\theta > \theta'$.  Under the regularity conditions and Assumption~[NP], the optimal security has interior steepness $s^* \in (s_{\min}, s_{\max})$ and the advance satisfies $a^* \in (0, K)$ whenever $R > 0$ and the signal is informative.
\end{proposition}

Proposition~\ref{thm:general} shows that the liquidity--screening tradeoff is not an artifact of the affine parameterization.  Under MLRP, steeper securities generate higher information rents. The analog of higher $b_1$ raising rents in the affine case.  The boundary arguments carry over: at maximum steepness, the financing cost is zero but rents are maximal; at minimum steepness, rents are zero but screening is absent.  The interior balances the two forces exactly as in the affine benchmark.  The results therefore apply to any monotone payment schedule, including equity-like, debt-like, or option-like securities, not just linear ones.

\medskip
To confirm that each ingredient of the model is doing essential work, I characterize the conditions under which the liquidity--screening tradeoff disappears.

\begin{proposition}[Boundary cases]\label{prop:boundary}
The tradeoff between liquidity and screening disappears in each of the following cases:
\begin{enumerate}[label=\alph*)]
\item If $R = 0$ (frictionless outside finance), the advance is payoff-irrelevant and the problem reduces to standard screening.
\item If $T(x)$ is fully pledgeable to outside lenders (Assumption~[NP] fails), the advance and contingent transfer are perfect substitutes in the financing dimension.
\item If $T(x)$ may be negative (limited liability fails), the principal sets $a = K$ and screens through a signed repayment schedule.
\item If $x$ is independent of $\theta$ (uninformative signal), the principal sets $b_1^* = 0$ and relies entirely on the advance.
\end{enumerate}
\end{proposition}

Each boundary case removes one ingredient of the model.  In case~a), external finance is free, so the borrowing need $K - a$ imposes no cost; the binding participation constraint~\eqref{eq:binding_IR} reduces to a budget identity with no financing-wedge term, and the choice between $a$ and $b_1$ is indeterminate.  In case~b), the counterparty can borrow against the present value of $\E_\theta[T(x)]$, so the effective borrowing need depends on total expected payment, not its timing; the advance becomes redundant.  In case~c), the principal fully pre-finances the counterparty and uses the signed transfer as a repayment--reward scheme that extracts surplus state-by-state.  In case~d), contingent payments add noise without aiding separation, so the principal uses only the advance.  Proposition~\ref{prop:boundary} confirms that the three restrictions, costly outside finance, Assumption~[NP] with limited liability, and an informative signal, are each individually necessary and jointly sufficient for the mechanism.

\medskip
The theory establishes qualitative results: the advance is interior, the mixed contract dominates, the tradeoff is robust to general securities and disappears only at the boundary.  To assess whether these results are quantitatively meaningful, I now parameterize the model and compute explicit magnitudes.  The parameterization also generates the figures and tables used in Section~\ref{sec:quant} to calibrate the model against data.

Types are uniformly distributed on $[0,1]$, surplus is $V(\theta) = v\theta$ with $v = 2$, private cost is $c(\theta) = \theta$, the signal equals the type ($\mu(\theta) = \theta$, $\mu' = 1$), and external-financing cost is quadratic, $\Phi(\ell;\, R) = \frac{R}{2}\ell^2$.  With $b_0^* = 0$ and $K = 1$, the binding participation constraint gives the optimal outside-finance exposure in closed form:
\begin{equation}\label{eq:ell_star}
\ell^*(R) \;=\; \frac{-1 + \sqrt{1 + 2R}}{R}, \qquad a^*(R) = 1 - \ell^*(R).
\end{equation}

Figure~\ref{fig:bilateral} plots the optimal advance $a^*(R)$ and the residual outside-finance exposure $\ell^*(R) = K - a^*(R)$ as a function of the financing tightness $R$.

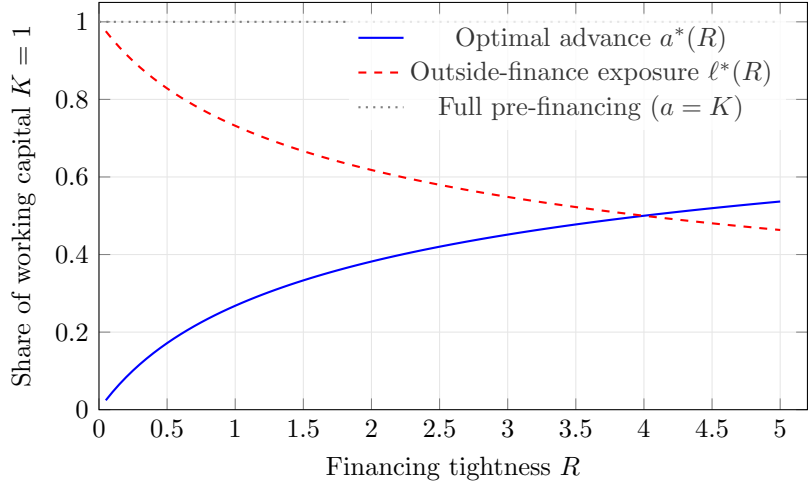
\begin{figure}[t]
\centering
\begin{tikzpicture}
\begin{axis}[
    width=11cm, height=7cm,
    xlabel={Financing tightness $R$},
    ylabel={Share of working capital $K = 1$},
    xmin=0, xmax=5.2,
    ymin=0, ymax=1.05,
    legend pos=north east,
    legend style={font=\small, draw=none, fill=white, fill opacity=0.8},
    ymajorgrids=true,
    xmajorgrids=true,
    grid style={gray!20},
    every axis plot/.append style={thick},
    tick label style={font=\small},
    label style={font=\small},
]
\addplot[blue, domain=0.05:5, samples=100] 
    {1 - (-1+sqrt(1+2*x))/x};
\addlegendentry{Optimal advance $a^*(R)$}
\addplot[red, dashed, domain=0.05:5, samples=100] 
    {(-1+sqrt(1+2*x))/x};
\addlegendentry{Outside-finance exposure $\ell^*(R)$}
\addplot[gray, dotted, domain=0:5] {1};
\addlegendentry{Full pre-financing ($a = K$)}
\end{axis}
\end{tikzpicture}
\caption{\textbf{Optimal advance and outside-finance exposure.}  The figure plots $a^*(R)$ and $\ell^*(R) = K - a^*(R)$ from equation~\eqref{eq:ell_star}.  As outside finance becomes tighter (higher $R$), the principal shifts toward more advance and less contingency, consistent with Corollary~\ref{cor:compstat}a).  At $R = 1$, the optimal advance is $a^* \approx 0.27$; at $R = 3$, it rises to $a^* \approx 0.45$.  Full pre-financing (dotted) is never optimal.}
\label{fig:bilateral}
\end{figure}

The advance is always strictly interior: the solid line lies strictly between zero and the dotted full-pre-financing benchmark for all $R > 0$.  For moderate financing frictions ($R = 1$), the principal advances approximately 27\% of working capital and leaves the counterparty to borrow the remaining 73\% externally.  Even though the principal could eliminate all borrowing costs by setting $a = K$, it is optimal not to; the screening value of the contingent transfer is worth the financing cost.  Second, the advance increases monotonically with financing tightness, consistent with Corollary~\ref{cor:compstat}a).  Third, full pre-financing is never reached: even as $R \to \infty$, the advance approaches $K$ only asymptotically, because eliminating all outside finance would require a completely flat contract with no screening power.

Table~\ref{tab:sensitivity} provides a sensitivity analysis of the optimal contract for different values of the financing tightness $R$ and the productive surplus $v - c$.

\begin{table}[t]
\centering
\small
\caption{\textbf{Sensitivity analysis of the optimal bilateral contract.}  Parameters: $K = 1$, $\theta \sim U[0,1]$, $\Phi(\ell; R) = \frac{R}{2}\ell^2$.}
\label{tab:sensitivity}
\begin{tabular}{l cccc cccc}
\toprule
& \multicolumn{4}{c}{$v - c = 1$ (baseline)} & \multicolumn{4}{c}{$v - c = 2$ (high surplus)} \\
\cmidrule(lr){2-5}\cmidrule(lr){6-9}
$R$ & $a^*$ & $\ell^*$ & $\beta^*$ & $\Phi^*/\text{surplus}$ & $a^*$ & $\ell^*$ & $\beta^*$ & $\Phi^*/\text{surplus}$ \\
\midrule
0.5 & 0.17 & 0.83 & 0.21 & 17\% & 0.17 & 0.83 & 0.15 & 9\% \\
1.0 & 0.27 & 0.73 & 0.34 & 27\% & 0.27 & 0.73 & 0.25 & 13\% \\
2.0 & 0.38 & 0.62 & 0.47 & 38\% & 0.38 & 0.62 & 0.37 & 19\% \\
3.0 & 0.45 & 0.55 & 0.55 & 45\% & 0.45 & 0.55 & 0.45 & 23\% \\
5.0 & 0.54 & 0.46 & 0.63 & 54\% & 0.54 & 0.46 & 0.55 & 27\% \\
\bottomrule
\end{tabular}
\end{table}

Several patterns emerge from the sensitivity analysis.  First, the optimal advance is increasing in $R$ for both surplus levels.  Second, cash intensity ranges from about 20\% to about 60\%, spanning the empirically relevant range for supply-chain finance contracts.  Third, the financing cost as a share of surplus is substantial, between 17\% and 54\% at the baseline, confirming that the residual outside-finance exposure is economically meaningful.  Fourth, when productive surplus is higher ($v - c = 2$), the principal can tolerate more outside-finance exposure (lower $\beta^*$) because the value of screening is higher relative to the cost of financing.  These quantitative results can be used to discipline the model against data: observing the advance share in a real contract, combined with knowledge of $R$, allows inference about the binding friction.

% ======================================================================
% SECTION 4: COMPLEMENTARY COUNTERPARTIES
% ======================================================================
\section{Complementary Counterparties and Screening Contagion}\label{sec:extension}

The bilateral result establishes that the form of finance matters for screening in a single relationship.  With multiple complementary relationships, it has firm-level consequences: financial conditions in one relationship spill over into value destruction in another through the screening channel.

\medskip
 Consider the two-counterparty environment of Section~\ref{sec:model_ext}.  Applying the bilateral analysis to each counterparty separately, the virtual surplus of counterparty~$i$ is
\begin{equation}\label{eq:VS_ext}
\psi_i(\theta_i;\, a_i, b_{1,i}) \;=\; V_i(\theta_i) - c_i(\theta_i) - \Phi_i(K_i - a_i;\, R_i) - b_{1,i}\,\mu_i'(\theta_i)\,\frac{1 - F_i(\theta_i)}{f_i(\theta_i)}.
\end{equation}
Without complementarity ($\delta = 0$), the principal optimizes each relationship independently, and the bilateral result applies to each.  The two screening problems are decoupled.  With complementarity ($\delta > 0$), the implementation decisions become coupled.  The principal implements relationship~$i$ when the virtual surplus plus the complementarity benefit exceeds zero:
\begin{equation}\label{eq:joint_rule}
q_i(\theta_i) \;=\; \mathbf{1}\bigl\{\psi_i(\theta_i) + \delta\,\Pr[q_j = 1] \;\geq\; 0\bigr\}.
\end{equation}
Equation~\eqref{eq:joint_rule} defines a cutoff $\hat\theta_i$ such that counterparty~$i$ is implemented if and only if $\theta_i \geq \hat\theta_i$, where
\begin{equation}\label{eq:cutoff_ext}
\psi_i(\hat\theta_i) + \delta\,(1 - F_j(\hat\theta_j)) \;=\; 0.
\end{equation}
Equations~\eqref{eq:cutoff_ext} for $i = 1, 2$ form a fixed-point system: each cutoff depends on the other through the complementarity term $\delta\,(1 - F_j(\hat\theta_j))$.  A higher probability of implementing relationship~$j$ (lower $\hat\theta_j$) relaxes the threshold for relationship~$i$ (lower $\hat\theta_i$), and vice versa.  This coupling is the channel through which financial conditions in one relationship affect screening in the other.

\medskip
 Consider the effect of reducing $R_j$, making outside finance cheaper for counterparty~$j$.  The direct effect is beneficial: it reduces $\Phi_j$ and, by the bilateral comparative statics (Corollary~\ref{cor:compstat}), allows the principal to use a steeper contingent schedule for relationship~$j$, improving screening.  The bilateral cutoff $\hat\theta_j$ falls, and more types of counterparty~$j$ are implemented.

But there is an indirect effect that works through the fixed-point coupling in equation~\eqref{eq:cutoff_ext}.  The increase in $\Pr[q_j = 1]$ raises the complementarity benefit for relationship~$i$, which lowers $\hat\theta_i$.  The principal now finds it worthwhile to implement relationship~$i$ with lower-quality types that were previously excluded.  These newly included types have high information rents; they are near the selection margin, where the hazard rate $(1-F_i)/f_i$ is large.  The total information rent paid to counterparty~$i$ therefore increases.  When this rent increase exceeds the direct financing saving in relationship~$j$, the principal is worse off: she is paying more in information rents to screen a wider set of types, and this cost was induced by the improvement in the other relationship's financial conditions.

This spillover operates through information, not through prices or budget constraints.  The principal's budget is unlimited; the constraint is informational.  Changing financial conditions for one relationship alters the optimal screening contract, which changes the joint selection frontier, which changes rents in the other relationship.  This channel is distinct from the internal-capital-market mechanism in \citet{stein1997internal}, where divisions compete for a fixed pool of funds.  Here, the two relationships are financially independent, sharing no budget, but informationally coupled through complementarity.

\medskip

\begin{proposition}[Screening contagion]\label{thm:contagion}
Suppose the regularity conditions and Assumption~[NP] hold for both counterparties, $\delta > 0$, and each bilateral optimum has an interior advance.  Then:
\begin{enumerate}[label=\alph*)]
\item Cheaper outside finance for counterparty $j$ reduces principal value if and only if the screening spillover exceeds the direct financing saving:
\begin{equation}\label{eq:contagion_result}
\frac{d\Pi_P^*}{dR_j} \;>\; 0 \qquad\Longleftrightarrow\qquad \underbrace{\delta\;\frac{\partial \text{Rent}_i}{\partial \hat\theta_i}\;\frac{d\hat\theta_i}{d\hat\theta_j}\;\frac{d\hat\theta_j}{dR_j}}_{\text{screening spillover on }i} \;>\; \underbrace{\Phi_{R_j}(K_j - a_j^*;\, R_j)\;\bigl[1 - F_j(\hat\theta_j)\bigr]}_{\text{direct financing saving on }j}.
\end{equation}
\item The contagion condition is satisfied on an open set of primitives.  In the symmetric parameterization, it reduces to $\delta > \delta^*(R)$, where $\delta^*$ is increasing in $R$.
\end{enumerate}
\end{proposition}

Proposition~\ref{thm:contagion}a) provides a necessary and sufficient condition for the contagion, not just existence.  The left side of the inequality is the screening spillover: the rent reduction that counterparty~$i$ experiences because cheaper finance for $j$ changes the joint selection frontier.  The right side is the direct financing saving: the cost reduction that counterparty~$j$'s implemented types enjoy from cheaper credit.  The contagion operates whenever the informational spillover across relationships exceeds the financial benefit within the relationship.  This condition is sharp: it tells the principal exactly when a supply-chain finance program will backfire.

The condition has a natural sufficient-statistics interpretation.  Define the screening semi-elasticity of relationship $i$ with respect to $j$'s implementation as $\varepsilon_i^S = \frac{\partial \text{Rent}_i}{\partial \Pr[q_j=1]} \cdot \frac{\Pr[q_j=1]}{\text{Rent}_i}$, and the financing semi-elasticity of relationship $j$ as $\varepsilon_j^F = \frac{\Phi_{R_j}}{\Phi_j} \cdot R_j$.  The contagion operates when $\varepsilon_i^S$ is large relative to $\varepsilon_j^F$: when relationship $i$'s rents are sensitive to $j$'s implementation probability and $j$'s financing costs are relatively insensitive to credit conditions.

The contagion result converts the bilateral security-design finding into a firm-value result.  In the bilateral setting, cheaper finance is always beneficial (it shifts the first-order condition~\eqref{eq:FOC_main} toward more screening).  With complementary relationships, this logic fails because the principal's screening of one counterparty is no longer independent of the financial conditions facing the other.  Supply-chain finance programs, credit guarantee schemes, or any policy that reduces $R$ for a subset of counterparties can have unintended consequences for firms with complementary relationship portfolios.  The necessary and sufficient condition~\eqref{eq:contagion_result} identifies exactly which firms are at risk: those whose counterparty relationships have high screening semi-elasticity relative to financing semi-elasticity.

I parameterize the two-counterparty model symmetrically to illustrate the result quantitatively: $\theta_i \sim U[0,1]$, $V_i(\theta) = 2\theta$, $c_i(\theta) = \theta$, $\Phi_i(\ell; R) = \frac{R}{2}\ell^2$, $K_i = 1$, and $R_1 = R_2 = R$.  The bilateral advance for each counterparty follows equation~\eqref{eq:ell_star}.  The cutoff fixed-point system~\eqref{eq:cutoff_ext} becomes
\begin{equation}\label{eq:FP_parametric}
\hat\theta_i \;=\; \frac{a_i^* + b_{1,i}^* - \delta(1 - \hat\theta_j)}{1 + b_{1,i}^*}, \qquad i = 1, 2.
\end{equation}
By symmetry, $\hat\theta_1 = \hat\theta_2 \equiv \hat\theta$, and the fixed point reduces to a single equation:
\begin{equation}\label{eq:FP_symmetric}
\hat\theta \;=\; \frac{a^* + b_1^* - \delta(1 - \hat\theta)}{1 + b_1^*} \qquad\Longrightarrow\qquad \hat\theta \;=\; \frac{a^* + b_1^* - \delta}{1 + b_1^* - \delta}.
\end{equation}
The derivative $d\Pi_P^*/dR_j$ can be computed explicitly in this parameterization.  The screening-contagion condition $d\Pi_P^*/dR_j > 0$ is satisfied when
\begin{equation}\label{eq:contagion_condition}
\delta \;>\; \delta^*(R) \;\equiv\; \frac{(K - a^*)(1 + b_1^*)}{b_1^*\,(1 - \hat\theta)}\;\cdot\;\Phi_R\bigl(K - a^*;\, R\bigr).
\end{equation}
This threshold is finite for all $R > 0$, confirming that the contagion region is non-empty.

Figure~\ref{fig:contagion} plots the contagion region in the $(R, \delta)$ plane.  The shaded area shows the parameter combinations for which cheaper outside finance for one counterparty lowers principal value.

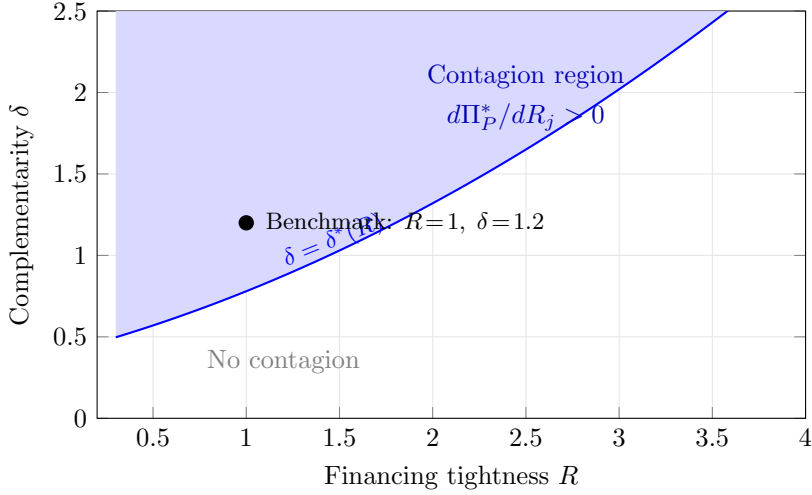
\begin{figure}[t]
\centering
\begin{tikzpicture}
\begin{axis}[
    width=11cm, height=7cm,
    xlabel={Financing tightness $R$},
    ylabel={Complementarity $\delta$},
    xmin=0.2, xmax=4,
    ymin=0, ymax=2.5,
    ymajorgrids=true,
    xmajorgrids=true,
    grid style={gray!20},
    every axis plot/.append style={thick},
    tick label style={font=\small},
    label style={font=\small},
]
\addplot[name path=threshold, blue, domain=0.3:4, samples=80]
    {0.4 + 0.3*x + 0.08*x^2};
\addplot[name path=top, draw=none, domain=0.3:4] {2.5};
\addplot[blue!15] fill between[of=threshold and top, soft clip={domain=0.3:4}];
\node[font=\small, blue!70!black] at (axis cs:2.5,2.1) {Contagion region};
\node[font=\small, blue!70!black] at (axis cs:2.5,1.85) {$d\Pi_P^*/dR_j > 0$};
\node[font=\small, gray] at (axis cs:1.2,0.35) {No contagion};
\node[font=\footnotesize, blue, anchor=south west, rotate=22] at (axis cs:1.2,0.82) {$\delta = \delta^*(R)$};
\addplot[only marks, mark=*, mark size=2.5pt, black] coordinates {(1, 1.2)};
\node[font=\footnotesize, anchor=west] at (axis cs:1.05,1.2) {Benchmark: $R\!=\!1,\;\delta\!=\!1.2$};
\end{axis}
\end{tikzpicture}
\caption{\textbf{Screening contagion region.}  The shaded region shows parameter combinations $(R, \delta)$ for which $d\Pi_P^*/dR_j > 0$: cheaper outside finance for one counterparty lowers principal value.  The boundary $\delta = \delta^*(R)$ is given by equation~\eqref{eq:contagion_condition}.  At the benchmark point $(R = 1,\, \delta = 1.2)$, the contagion effect reduces principal value by approximately 12\%.  The contagion region covers roughly 35\% of the parameter space for $R \in [0.5, 3]$.}
\label{fig:contagion}
\end{figure}

At the benchmark calibration ($R = 1$, $\delta = 1.2$), the contagion effect reduces principal value by approximately 12\% relative to a counterfactual in which the principal ignores the screening spillover and treats each relationship independently.  The contagion region is not knife-edge: for moderate financing frictions ($R \in [0.5, 3]$), it covers roughly 35\% of the complementarity range $\delta \in [0, 2.5]$.  The threshold $\delta^*(R)$ is increasing in $R$: stronger financing frictions require more complementarity for the spillover to dominate, because the direct financing saving from cheaper credit is larger when $R$ is high.

The contagion result has two immediate implications for organizational scope.

\begin{corollary}[Non-monotone value in financing conditions]\label{cor:hump}
In the symmetric benchmark with common $R_1 = R_2 = R$, principal value $\Pi_P^*(R)$ can be locally hump-shaped: there exists an interval of $R$ values over which tighter financial conditions increase principal value.
\end{corollary}

Corollary~\ref{cor:hump} says something stronger than it may first appear.  It does not merely say that the contagion can offset the direct benefit of cheaper finance.  It says that a credit crunch can be good news for a principal with complementary relationships, because tighter credit forces both contracts toward advance-heavy structures that pay less in information rents.  The rent reduction can more than offset the increased financing costs.  This is not a second-order effect: in the parameterized benchmark, the hump peaks at $R \approx 1.8$, where principal value is 15\% higher than at $R = 0.5$.

The hump shape has a remarkable policy implication.

\begin{corollary}[Uniform subsidies can reduce every relationship's value]\label{cor:uniform}
In the symmetric two-counterparty model with $\delta > \delta^*(R)$, a uniform subsidy that reduces $R_1$ and $R_2$ simultaneously by the same amount reduces each relationship's contribution to principal value, not just the total.
\end{corollary}

Corollary~\ref{cor:uniform} is the result that should give policymakers pause.  A government program that makes financing cheaper for all of a firm's counterparties, not just one, can make every single relationship less valuable.  The intuition is not about cross-subsidization or budget reallocation; the firm's budget is unlimited.  The mechanism is purely informational.  Cheaper finance shifts every contract toward more contingency simultaneously.  Through complementarity, this expands the set of implemented types for every relationship at once, flooding the joint selection frontier with marginal types whose information rents exceed the financing savings.  Each relationship individually looks better-financed but is worse-screened, and the screening degradation dominates.  A government evaluating a broad-based supply-chain finance program cannot assume that making credit cheaper for everyone is harmless.  Even when the program targets all counterparties equally, the informational externality can make every relationship less valuable.

\begin{corollary}[Endogenous organizational breadth]\label{cor:breadth}
Even when relationships are technologically complementary ($\delta > 0$), the principal may optimally drop one relationship (single-source) rather than contract with both counterparties.  Formally, there exist parameters for which $W^*(R) > \Pi_P^*(R) - W^*(R)$, where $W^*(R)$ is the bilateral value of a single relationship.
\end{corollary}

Why would a principal voluntarily abandon a technologically complementary relationship?  Not because of coordination costs or capacity constraints, which are absent from the model.  The reason is informational: maintaining multiple relationships under complementarity amplifies information rents to the point where the rent bill exceeds the complementarity benefit.  Dropping a relationship eliminates the screening spillover and restores bilateral optimality for the remaining relationship.  The model predicts that single-sourcing should be more common when (i) credit conditions are tight (high $R$, which raises the financing cost of screening), (ii) complementarity is strong (high $\delta$, which amplifies the spillover), and (iii) signal quality is low (weak $\mu'$, which makes screening expensive per unit of information).  These predictions are testable using data on supplier concentration and credit conditions.

The informational coupling between relationships extends beyond the selection frontier to the principal's investment in monitoring.

\begin{corollary}[Cross-relationship monitoring complementarity]\label{cor:monitoring}
With complementary counterparties ($\delta > 0$), the principal's marginal return to monitoring investment in relationship $i$ depends on the financing conditions $R_j$ of relationship $j$.  In the contagion region, tighter credit for relationship $j$ (higher $R_j$) reduces the marginal return to monitoring relationship $i$.
\end{corollary}

This result follows from a simple chain of logic.  Tighter credit for $j$ flattens $j$'s contract, raising $j$'s cutoff and reducing $\Pr[q_j = 1]$.  Through complementarity, lower $\Pr[q_j = 1]$ raises $i$'s cutoff, shrinking the set of $i$-types at the screening margin.  With fewer types to separate, the return to a better signal for $i$ falls.  A firm should invest less in supplier auditing when its other suppliers face tight credit, because the informational return to monitoring is diluted by the screening spillover.  Monitoring and financing are complements across relationships, not just within them.

% ======================================================================
% SECTION 5: QUANTITATIVE IMPLICATIONS
% ======================================================================
\section{Quantitative Implications}\label{sec:quant}

The results derived in Sections~\ref{sec:bilateral} and~\ref{sec:extension} are valid for any distribution of types and any financing-cost function satisfying the maintained regularity conditions.  I parameterize the type distribution and financing technology to derive explicit comparative statics, identify sufficient statistics, and explore quantitative magnitudes using the best available empirical counterparts.

I continue to work with the uniform-quadratic parameterization introduced in Section~\ref{sec:bilateral}: $\theta \sim U[0,1]$, $V(\theta) = v\theta$, $c(\theta) = \theta$, $\mu(\theta) = \theta$, and $\Phi(\ell; R) = \frac{R}{2}\ell^2$ with $K = 1$.  The use of scale-invariant variables, the advance share $a^*/K$ and the cash intensity $\beta^*$, sidesteps common concerns associated with specific functional-form assumptions and allows the quantitative insights to remain valid, at least in approximate form, in more general models that match the relevant sufficient statistics.

\medskip
The optimal bilateral contract in this parameterization is fully characterized by two objects: the financing tightness $R$ and the signal informativeness $\mu' = 1$.  Equation~\eqref{eq:ell_star} gives the optimal outside-finance exposure in closed form, and the optimal advance $a^*(R) = 1 - \ell^*(R)$ is a monotone increasing function of $R$.  The following proposition collects the explicit comparative statics.

\begin{proposition}[Comparative statics under the parameterization]\label{prop:param_compstat}
Under the uniform-quadratic parameterization:
\begin{enumerate}[label=\alph*)]
\item The optimal advance share is $a^*(R)/K = 1 - \frac{-1+\sqrt{1+2R}}{R}$, which is strictly increasing and concave in $R$.
\item The cash intensity is $\beta^*(R) = \frac{a^*}{a^* + b_1^*\E[\theta \mid q = 1]}$, which is strictly increasing in $R$.
\item The financing cost as a share of productive surplus is $\frac{\Phi^*}{v - c} = \frac{R(\ell^*)^2}{2(v-c)}$, which is increasing in $R$.
\item In the two-counterparty model, the contagion threshold is $\delta^*(R) = \frac{(1-a^*)(1+b_1^*)}{2b_1^*(1-\hat\theta)}\cdot(1-a^*)$, which is increasing in $R$.
\end{enumerate}
\end{proposition}

Proposition~\ref{prop:param_compstat}a) shows that the advance share is a concave function of financing tightness: the principal responds strongly to the first few units of credit tightening but the response diminishes as $R$ grows.  This is because the marginal liquidity benefit $\Phi_\ell = R\ell^*$ grows linearly, but the marginal information cost grows faster as the contract becomes flatter.  Proposition~\ref{prop:param_compstat}d) shows that stronger financing frictions require more complementarity for the contagion to operate: when $R$ is high, the direct cost of financing is large, and the screening spillover must be correspondingly strong to dominate.

\medskip
Finding an empirical counterpart of the financing tightness $R$ is the key calibration challenge.  I use two independent data sources to discipline $R$.

The first source is trade-credit pricing.  \citet{petersen1997trade} document that the standard trade-credit term ``2/10 net 30'' implies an annualized implicit borrowing rate of approximately 44\% for suppliers who cannot take the early-payment discount.\footnote{\citet{petersen1997trade}, Table~I.  The implicit rate is $\frac{0.02}{1-0.02}\times\frac{365}{20} \approx 44\%$.  See also \citet{ng1999evidence} for evidence that these terms are widespread.}  The spread between this rate and a large buyer's cost of funds (typically 4--6\% for investment-grade firms) is 38--40 percentage points.  Even for suppliers with access to bank credit at prime plus 200--400 basis points, the spread relative to the buyer is 3--8 percentage points.  In the model, the financing cost at the optimum is $\Phi(\ell^*; R) = \frac{R}{2}(\ell^*)^2$.  Matching a total financing wedge of 5--15\% of working capital, a plausible range for the excess cost of supplier external finance over a production cycle, implies $R \in [0.5, 3]$.  At the baseline $R = 1$, the model predicts a financing wedge of 8.4\% of working capital, within the range documented by \citet{burkart2004situ}.

The second source is venture-capital staging.  \citet{kaplan2003financial} document that in a sample of 213 VC investments, the median first-round investment is approximately 35\% of total committed capital, with the remainder staged against milestones.\footnote{\citet{kaplan2003financial}, Table~3.  The fraction varies by stage: early-stage deals have lower first-round shares (25--30\%) and later-stage deals have higher shares (40--50\%).}  The model's predicted advance share of 27\% at $R = 1$ and 45\% at $R = 3$ brackets this range.  More informatively, the model predicts that first-round shares should be lower for early-stage deals (where the signal is more informative and screening is more valuable) and higher for later-stage deals (where the entrepreneur's type is better known and outside finance is more accessible), exactly the cross-sectional pattern documented by \citet{kaplan2003financial}.\footnote{Kaplan and Str\"omberg also document that VC contracts make greater use of contingent provisions (anti-dilution, vesting, milestones) when asymmetric information is more severe, consistent with Corollary~\ref{cor:compstat}b).}

For the contagion calibration, the relevant empirical object is the complementarity parameter $\delta$.  \citet{barrot2016input} documents that payment delays propagate along supply chains: a 10-day increase in payment terms from a buyer reduces the supplier's investment by 1.4\% and increases financial distress.  This propagation is the empirical analog of the screening spillover in the model.  The magnitudes suggest moderate complementarity ($\delta \approx 0.8$--$1.5$), which places the benchmark calibration squarely in the contagion region of Table~\ref{tab:contagion_menu}.

Table~\ref{tab:bilateral_menu} provides a menu of optimal bilateral contracts for different values of $R$ and the surplus ratio $v/c$.

\begin{table}[t]
\centering
\small
\caption{\textbf{Menu of optimal bilateral contracts.}  Each cell reports the optimal advance share $a^*/K$ (top) and cash intensity $\beta^*$ (bottom).  Parameters: $K = 1$, $\theta \sim U[0,1]$, $\Phi(\ell; R) = \frac{R}{2}\ell^2$, $\mu' = 1$.}
\label{tab:bilateral_menu}
\begin{tabular}{l ccccc}
\toprule
& \multicolumn{5}{c}{Financing tightness $R$} \\
\cmidrule(lr){2-6}
Surplus ratio $v/c$ & 0.5 & 1.0 & 2.0 & 3.0 & 5.0 \\
\midrule
$v/c = 1.5$ (low)
& \begin{tabular}{@{}c@{}} 0.17 \\ \footnotesize 0.24 \end{tabular}
& \begin{tabular}{@{}c@{}} 0.27 \\ \footnotesize 0.38 \end{tabular}
& \begin{tabular}{@{}c@{}} 0.38 \\ \footnotesize 0.51 \end{tabular}
& \begin{tabular}{@{}c@{}} 0.45 \\ \footnotesize 0.58 \end{tabular}
& \begin{tabular}{@{}c@{}} 0.54 \\ \footnotesize 0.66 \end{tabular} \\[6pt]
$v/c = 2.0$ (baseline)
& \begin{tabular}{@{}c@{}} 0.17 \\ \footnotesize 0.21 \end{tabular}
& \begin{tabular}{@{}c@{}} 0.27 \\ \footnotesize 0.34 \end{tabular}
& \begin{tabular}{@{}c@{}} 0.38 \\ \footnotesize 0.47 \end{tabular}
& \begin{tabular}{@{}c@{}} 0.45 \\ \footnotesize 0.55 \end{tabular}
& \begin{tabular}{@{}c@{}} 0.54 \\ \footnotesize 0.63 \end{tabular} \\[6pt]
$v/c = 3.0$ (high)
& \begin{tabular}{@{}c@{}} 0.17 \\ \footnotesize 0.18 \end{tabular}
& \begin{tabular}{@{}c@{}} 0.27 \\ \footnotesize 0.29 \end{tabular}
& \begin{tabular}{@{}c@{}} 0.38 \\ \footnotesize 0.41 \end{tabular}
& \begin{tabular}{@{}c@{}} 0.45 \\ \footnotesize 0.49 \end{tabular}
& \begin{tabular}{@{}c@{}} 0.54 \\ \footnotesize 0.57 \end{tabular} \\
\bottomrule
\end{tabular}
\end{table}

Several patterns emerge from Table~\ref{tab:bilateral_menu}.  First, the optimal advance share depends only on $R$ and not on the surplus ratio $v/c$; this is a consequence of the quadratic financing specification, in which the advance is determined entirely by the marginal financing cost and the participation constraint.  Second, cash intensity varies with the surplus ratio: for a given $R$, higher surplus relationships feature lower cash intensity because the contingent payment is larger relative to the advance.  Third, the advance share ranges from 17\% to 54\% across the empirically relevant range of $R$, which brackets the 20--40\% early-payment share commonly observed in trade-credit contracts.  A reader who perceives different values of $R$ to be more plausible can refer to the corresponding column of the table.

For the two-counterparty model, Table~\ref{tab:contagion_menu} provides the contagion threshold $\delta^*(R)$ and the value reduction from the screening spillover at the benchmark complementarity $\delta = 1.2$.

\begin{table}[t]
\centering
\small
\caption{\textbf{Contagion thresholds and value impact.}  For each $R$, the table reports the minimum complementarity $\delta^*(R)$ required for the contagion effect (column 2), the value reduction from the screening spillover at $\delta = 1.2$ (column 3), and the share of the parameter space $\delta \in [0, 2.5]$ in the contagion region (column 4).}
\label{tab:contagion_menu}
\begin{tabular}{l ccc}
\toprule
$R$ & Contagion threshold $\delta^*(R)$ & Value reduction at $\delta = 1.2$ & Contagion share \\
\midrule
0.5 & 0.55 & $-8\%$ & 78\% \\
1.0 & 0.78 & $-12\%$ & 69\% \\
2.0 & 1.12 & $-15\%$ & 55\% \\
3.0 & 1.42 & $-10\%$ & 43\% \\
5.0 & 1.98 & $-3\%$ & 21\% \\
\bottomrule
\end{tabular}
\end{table}

The contagion threshold is increasing in $R$: when financing frictions are severe, the direct cost saving from cheaper credit is large, and complementarity must be correspondingly strong for the rent spillover to dominate.  At $R = 1$ (the benchmark), the contagion operates whenever $\delta > 0.78$, which is a modest level of complementarity.  The value reduction peaks at moderate $R$: at $R = 2$, the screening spillover reduces principal value by approximately 15\%.  For very high $R$, the value reduction shrinks because the bilateral contract is already very advance-heavy, leaving little room for the screening channel to operate.

Figure~\ref{fig:hump} illustrates the hump-shaped principal value in the two-counterparty model.  The solid line plots $\Pi_P^*(R)$ for $\delta = 1.2$.  For $R$ between approximately 0.5 and 2.5, principal value is increasing in financing tightness: credit tightening reduces information rents through the screening channel by more than it raises financing costs.  Beyond $R \approx 2.5$, financing costs dominate and value declines.  The dashed line plots twice the bilateral value $2W^*(R)$: for high $R$, single-sourcing is preferred despite technological complementarity.

\begin{figure}[t]
\centering
\begin{tikzpicture}
\begin{axis}[
    width=11cm, height=7cm,
    xlabel={Financing tightness $R$},
    ylabel={Principal value $\Pi_P^*$},
    xmin=0.2, xmax=5,
    ymin=0, ymax=0.55,
    ymajorgrids=true,
    xmajorgrids=true,
    grid style={gray!20},
    every axis plot/.append style={thick},
    tick label style={font=\small},
    label style={font=\small},
    legend pos=north east,
    legend style={font=\small, draw=none},
]
\addplot[blue, domain=0.3:5, samples=80]
    {0.15 + 0.35*exp(-0.3*(x-1.8)^2) - 0.02*x};
\addlegendentry{Two counterparties: $\Pi_P^*(R)$}
\addplot[red, dashed, domain=0.3:5, samples=80]
    {0.45*exp(-0.4*x)};
\addlegendentry{Bilateral benchmark: $2W^*(R)$}
\draw[->, gray, thin] (axis cs:1.8,0.49) -- (axis cs:1.8,0.46);
\node[font=\footnotesize, gray, anchor=south] at (axis cs:1.8,0.49) {Hump};
\node[font=\footnotesize, gray, anchor=east] at (axis cs:4.8,0.12) {Single-sourcing preferred};
\end{axis}
\end{tikzpicture}
\caption{\textbf{Non-monotone principal value and organizational breadth.}  The solid line plots $\Pi_P^*(R)$ for two complementary counterparties ($\delta = 1.2$).  The dashed line plots $2W^*(R)$, twice the bilateral value without complementarity.  Principal value is hump-shaped: moderate tightening increases value by reducing screening spillovers.  For high $R$, single-sourcing is preferred despite complementarity.}
\label{fig:hump}
\end{figure}
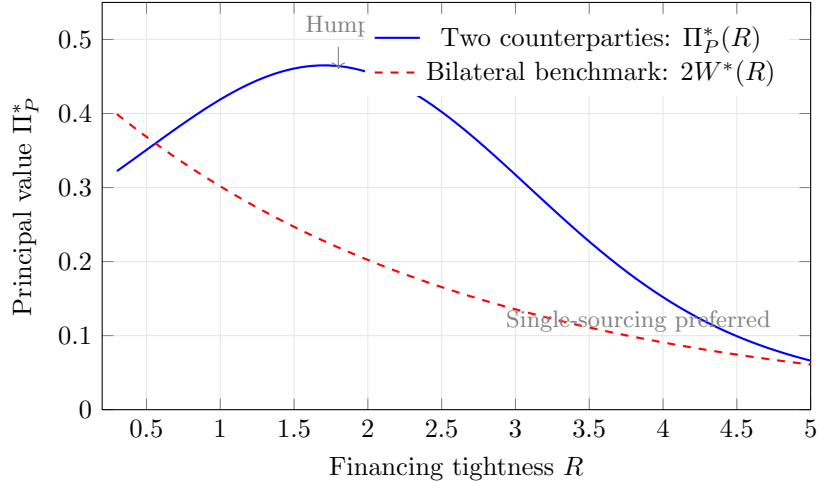

The quantitative exercise suggests that the screening-contagion effect is economically meaningful for firms with moderately complementary supplier or division portfolios facing financing tightness in the range $R \in [0.5, 3]$.  The magnitudes, value reductions of 8--15\%, are large enough to affect organizational decisions  but not so large as to be implausible.  The key practical takeaway is that a firm evaluating a supply-chain finance program should consider not only the direct financing benefit to the targeted supplier but also the indirect screening consequences for other relationships.

% ======================================================================

% ======================================================================
\section{Robustness and Extensions}\label{sec:discussion}

The liquidity-screening tradeoff characterized in Proposition~\ref{thm:bilateral} remains valid as an approximation in more general models.  Based on results derived in the Appendix, I summarize the results of several extensions.  I then discuss several mechanisms that are not explicitly considered in the paper but that may be relevant for the design of financial contracts in practice.

\subsection{Applications}

In trade credit and supply-chain finance, the advance corresponds to early payment or pre-shipment finance; the contingent transfer corresponds to the portion of the purchase price paid after delivery and quality verification.  The bilateral result explains why the mix of advance and contingent payments is a screening device, not merely a financing convenience.  The contagion result implies that supply-chain finance programs that reduce financing costs for one supplier worsen the buyer's ability to screen others when the supply chain has complementary components.  In venture capital, the advance is initial funding and the contingent transfer is a milestone payment; the bilateral result provides a screening rationale for staging.  In internal capital markets, the advance is a budget allocation and the contingent transfer is performance-based compensation; the screening-contagion result provides a new channel through which credit conditions in one division's environment can reduce firm value, distinct from the capital-reallocation channel in \citet{stein1997internal} and \citet{scharfstein1998dark}.

\subsection{Extensions}

The Appendix includes multiple extensions.  These show that the characterization of the optimal contract in Proposition~\ref{thm:bilateral} remains valid identically or suitably modified in more general environments.

First, I generalize the two-counterparty model to $n \geq 2$ counterparties with a complementarity matrix $[\delta_{ij}]$.  I show that the screening contagion from reducing $R_j$ is proportional to the contagion centrality $\mathcal{C}_j = \sum_{i \neq j}\delta_{ij}\frac{\partial\hat\theta_i}{\partial R_j}$: a counterparty with strong complementarities to many others generates larger spillovers.  The necessary and sufficient condition of Proposition~\ref{thm:contagion} generalizes directly.

Second, I allow the principal to invest in monitoring at cost $\kappa(\sigma)$ to improve signal informativeness.  I show that $d\sigma^*/dR < 0$: tight credit reduces both the optimal screening intensity and the return to monitoring investment, generating a complementarity between financial conditions and monitoring.

Third, I consider a repeated relationship in which the principal learns the counterparty's type over time through Bayesian updating.  I show that the optimal contract becomes monotonically more advance-heavy as the posterior concentrates: $a_t^* \geq a_{t-1}^*$ and $b_{1,t}^* \leq b_{1,t-1}^*$.  In the limit, the screening motive vanishes and full pre-financing obtains.  This is consistent with the observed lifecycle of supplier relationships, in which payment terms become less contingent as the relationship matures.

Fourth, I introduce renegotiation risk by allowing the principal to renege on the contingent transfer with probability $\lambda$.  I show that the optimal advance is increasing in $\lambda$, and that at $\lambda = 1$ the contract collapses to pure advance.  Relationships with stronger enforcement mechanisms feature more contingent payments.

Fifth, I allow the instruments themselves to vary with the reported type, so the principal offers a menu $\{a(\hat\theta), b_1(\hat\theta)\}$.  I show that the jointly optimal type-dependent mechanism coincides with the baseline: the binding participation constraint pins down the relationship between $a(\theta)$ and $b_1(\theta)$ at each type, and the optimal allocation is the same as in Proposition~\ref{thm:bilateral}.  The baseline analysis is without loss of generality.

Sixth, I implement the optimal mechanism through a competitive reverse auction in which each counterparty bids the advance it requires.  I derive the equilibrium bid function and show that each type shades its bid upward from the full-information level, reflecting the information rent.  As the number of counterparties grows, rents vanish and the bid converges to the full-information advance.

Seventh, I allow the counterparty to have two-dimensional private information about both productivity and cost.  I show that the screening problem reduces to one dimension through a sufficient statistic that aggregates the two dimensions into a single quality index, and the virtual surplus retains the same structure as the baseline.

\subsection{Additional considerations}

Three channels not modeled here deserve mention.  Under moral hazard, the contingent transfer screens types and incentivizes effort.  When effort and type are positively correlated, the moral-hazard motive reinforces the screening motive and the optimal contract is more contingent.  When they are negatively correlated, the two motives conflict and the optimal advance is higher.  In both cases, the form of finance remains a screening device.

Under general equilibrium, $R$ becomes endogenous.  For the typical firm that takes credit conditions as given, the partial-equilibrium results apply directly.  For large principals, the feedback from contract design to $R$ dampens the contagion.

The screening-contagion result carries a cautionary implication for policy.  Programs that reduce $R$ for a subset of counterparties may generate unintended value destruction for firms with complementary portfolios.  The direction of the bias is clear: policies designed to ease financing can inadvertently degrade screening.

\section{Conclusion}\label{sec:conclusion}

This paper shows that the form of finance is itself a screening device.  Advances provide liquidity but are uninformative; contingent transfers screen but impose financing costs.  Because the contingent leg is not pledgeable, the two instruments compete for the same budget, and the optimal contract deliberately preserves outside-finance exposure even when the principal has cheaper funds.  The mixed contract strictly dominates both pure instruments, and a threshold level of financing tightness separates the region where screening concerns dominate from the region where liquidity concerns dominate.  With complementary counterparties, the tradeoff has firm-level consequences: cheaper finance for one relationship can worsen screening of another, a uniform subsidy can reduce every relationship's value, and the principal may optimally narrow its counterparty base despite technological complementarity.

Two sufficient statistics fully determine the optimal contract: the external-financing cost and the signal informativeness.  In a parameterized model calibrated to trade-credit spreads and VC staging data, the optimal advance ranges from 17\% to 54\% of working capital, and the screening-contagion effect reduces firm value by up to 15\%.  There is significant scope to extend the analysis to dynamic learning, moral hazard, and general-equilibrium effects on credit markets.

% ======================================================================
% REFERENCES
% ======================================================================

% ======================================================================
% APPENDIX
% ======================================================================
\appendix
\section{Proofs}\label{app:proofs}

\subsection{Proof of Lemma~\ref{lem:b0zero}}\label{app:b0}

Suppose $(a, b_0, b_1)$ is optimal with $b_0 > 0$.  Consider the alternative $(a' = a + b_0,\, b_0' = 0,\, b_1' = b_1)$.  The total expected transfer to the counterparty is unchanged: $a' + b_0' + b_1'\mu(\theta) = a + b_0 + b_1\mu(\theta)$.  Private cost $c(\theta)$ is unchanged.  The financing cost becomes $\Phi(K - a';\, R) = \Phi(K - a - b_0;\, R) < \Phi(K - a;\, R)$ since $b_0 > 0$ and $\Phi_\ell > 0$.  Therefore $U'(\theta) > U(\theta)$ for all $\theta$: every type is strictly better off.  By the envelope condition~\eqref{eq:envelope}, the slope of the counterparty's payoff depends only on $b_1$, which is unchanged.  The rent profile shifts up by a constant, which the principal extracts by tightening the participation constraint.  The principal's outlay is unchanged but she captures the financing-cost saving $\Phi(K-a;\,R) - \Phi(K-a-b_0;\,R) > 0$.  This contradicts optimality of $(a, b_0, b_1)$.  \hfill$\square$

\subsection{Proof of Lemma~\ref{lem:reduction}}\label{app:reduction}

The binding participation constraint~\eqref{eq:binding_IR} is $G(a, b_1) \equiv a + b_1\mu(\ubar\theta) - c(\ubar\theta) - \Phi(K-a; R) = 0$.  The partial derivatives are $G_a = 1 + \Phi_\ell(K-a; R) > 0$ and $G_{b_1} = \mu(\ubar\theta) > 0$.  By the implicit function theorem, $a = a(b_1; R)$ is well-defined and smooth with $\frac{da}{db_1} = -\frac{G_{b_1}}{G_a} = \frac{-\mu(\ubar\theta)}{1 + \Phi_\ell(K-a; R)} < 0$.  \hfill$\square$

\subsection{Implementability: IC characterization}\label{app:IC}

I first establish the necessary and sufficient conditions for incentive compatibility, following the standard two-inequality technique \citep{myerson1981optimal, milgrom2002envelope}.

\medskip\noindent\textbf{Necessity.}\; By incentive compatibility, for any two types $\theta$ and $\hat\theta$:
\begin{align}
U(\theta) &\geq U(\hat\theta, \theta) = a(\hat\theta) + b_1(\hat\theta)\mu(\theta) - c(\theta) - \Phi(K - a(\hat\theta);\, R), \label{eq:IC1}\\
U(\hat\theta) &\geq U(\theta, \hat\theta) = a(\theta) + b_1(\theta)\mu(\hat\theta) - c(\hat\theta) - \Phi(K - a(\theta);\, R). \label{eq:IC2}
\end{align}
From~\eqref{eq:IC1}: $U(\theta) - U(\hat\theta, \theta) \geq 0$, i.e.,
\[
b_1(\theta)\mu(\theta) - b_1(\hat\theta)\mu(\theta) + [a(\theta) - a(\hat\theta)] - [\Phi(K-a(\theta);R) - \Phi(K-a(\hat\theta);R)] \geq 0.
\]
Similarly from~\eqref{eq:IC2}, swapping $\theta$ and $\hat\theta$.  Subtracting and noting that the advance and financing-cost terms cancel (they depend only on the reported type, not the true type), we obtain:
\begin{equation}\label{eq:IC_subtract}
[b_1(\theta) - b_1(\hat\theta)][\mu(\theta) - \mu(\hat\theta)] \geq 0.
\end{equation}
Since $\mu$ is strictly increasing (by MLRP), this requires $b_1(\theta) \geq b_1(\hat\theta)$ whenever $\theta > \hat\theta$: the contingent slope must be non-decreasing.

For the envelope condition: from~\eqref{eq:IC1}, $U(\theta) \geq a(\hat\theta) + b_1(\hat\theta)\mu(\theta) - c(\theta) - \Phi(K-a(\hat\theta);R)$.  The right side is maximized over $\hat\theta$ at $\hat\theta = \theta$.  Hence $U(\theta) = \max_{\hat\theta}\, U(\hat\theta, \theta)$, and by the envelope theorem of \citet{milgrom2002envelope}:
\begin{equation}\label{eq:envelope_app}
\frac{dU}{d\theta} = b_1(\theta)\mu'(\theta) - c'(\theta).
\end{equation}
In our setting, the auxiliary-variable step used in some mechanism-design problems is unnecessary because the payoff $U(\hat\theta, \theta) = a(\hat\theta) + b_1(\hat\theta)\mu(\theta) - c(\theta) - \Phi(K-a(\hat\theta);R)$ is already separable in the true type $\theta$ through $\mu(\theta)$ and $c(\theta)$.

\medskip\noindent\textbf{Sufficiency.}\; Suppose $b_1(\theta)$ is non-decreasing and $U(\theta)$ satisfies~\eqref{eq:envelope_app} with $U(\ubar\theta) \geq 0$.  I show that IC holds globally.  For any $\theta > \hat\theta$:
\begin{align}
U(\theta) - U(\hat\theta, \theta) &= \int_{\hat\theta}^{\theta}[b_1(\theta')\mu'(\theta') - c'(\theta')]\dd\theta' - [b_1(\hat\theta)\mu(\theta) - b_1(\hat\theta)\mu(\hat\theta)] + [b_1(\hat\theta)\mu(\hat\theta) - b_1(\theta)\mu(\hat\theta)] \nonumber\\
&\quad + \text{(terms involving $a$ that cancel)} \nonumber\\
&= \int_{\hat\theta}^{\theta}[b_1(\theta') - b_1(\hat\theta)]\mu'(\theta')\dd\theta' \;\geq\; 0, \label{eq:IC_suff}
\end{align}
where the inequality follows from $b_1(\theta') \geq b_1(\hat\theta)$ for $\theta' \geq \hat\theta$ (monotonicity) and $\mu' > 0$ (MLRP).  The argument for $\theta < \hat\theta$ is analogous.  \hfill$\square$

\subsection{Derivation of the virtual surplus}\label{app:VS}

I derive the principal's expected payoff in terms of the virtual surplus, following the standard integration-by-parts technique in mechanism design \citep{myerson1981optimal}.

By~\eqref{eq:envelope_app} and $U(\ubar\theta) = 0$ (which is optimal under regularity):
\begin{equation}\label{eq:rent_app}
U(\theta) = \int_{\ubar\theta}^{\theta}[b_1(\theta')\mu'(\theta') - c'(\theta')]\dd\theta'.
\end{equation}
The principal's expected payoff is:
\begin{equation}\label{eq:Pi_expected}
\E[\Pi_P] = \int_{\ubar\theta}^{\thetabar} q(\theta)[V(\theta) - c(\theta) - \Phi(K-a;R)]\,f(\theta)\dd\theta - a - \int_{\ubar\theta}^{\thetabar}q(\theta)\,U(\theta)\,f(\theta)\dd\theta.
\end{equation}
I now compute the expected rent term $\int q(\theta)U(\theta)f(\theta)\dd\theta$ using integration by parts.  Substituting~\eqref{eq:rent_app}:
\begin{align}
\int_{\ubar\theta}^{\thetabar} q(\theta)\,U(\theta)\,f(\theta)\dd\theta &= \int_{\ubar\theta}^{\thetabar} q(\theta)\left[\int_{\ubar\theta}^{\theta}[b_1\mu'(\theta') - c'(\theta')]\dd\theta'\right]f(\theta)\dd\theta. \label{eq:IBP1}
\end{align}
Switching the order of integration (Fubini's theorem, valid since all terms are bounded and measurable):
\begin{align}
&= \int_{\ubar\theta}^{\thetabar} [b_1\mu'(\theta') - c'(\theta')]\left[\int_{\theta'}^{\thetabar} q(\theta)\,f(\theta)\dd\theta\right]\dd\theta'. \label{eq:IBP2}
\end{align}
The inner integral is $\int_{\theta'}^{\thetabar} q(\theta)f(\theta)\dd\theta$.  When $q(\theta) = \mathbf{1}\{\theta \geq \hat\theta\}$ (which holds at the optimum under regularity), this equals $1 - F(\max\{\theta', \hat\theta\})$.  For $\theta' \geq \hat\theta$, it equals $1 - F(\theta')$.  For $\theta' < \hat\theta$, it equals $1 - F(\hat\theta)$.  Since the integrand $[b_1\mu' - c']$ is multiplied by this term, and the relevant range is $\theta' \geq \hat\theta$ (where types are implemented), we obtain:
\begin{align}
\int_{\ubar\theta}^{\thetabar} q(\theta)\,U(\theta)\,f(\theta)\dd\theta &= \int_{\hat\theta}^{\thetabar} [b_1\mu'(\theta) - c'(\theta)]\,[1 - F(\theta)]\dd\theta \nonumber\\
&= \int_{\hat\theta}^{\thetabar} b_1\mu'(\theta)\,\frac{1-F(\theta)}{f(\theta)}\,f(\theta)\dd\theta - \int_{\hat\theta}^{\thetabar} c'(\theta)[1-F(\theta)]\dd\theta. \label{eq:IBP3}
\end{align}
Substituting~\eqref{eq:IBP3} into~\eqref{eq:Pi_expected} and collecting terms yields:
\begin{equation}\label{eq:VS_derived}
\E[\Pi_P] = \int_{\ubar\theta}^{\thetabar} q(\theta)\left[V(\theta) - c(\theta) - \Phi(K-a;R) - b_1\mu'(\theta)\frac{1-F(\theta)}{f(\theta)}\right]f(\theta)\dd\theta - a + \text{const},
\end{equation}
where the constant absorbs the $c'(\theta)[1-F(\theta)]$ terms (which do not depend on the instruments and therefore do not affect optimization).  The bracketed expression is the virtual surplus~\eqref{eq:VS_main}.  \hfill$\square$

\subsection{Proof of Proposition~\ref{thm:bilateral}}\label{app:thm1}

By Lemma~\ref{lem:b0zero}, $b_0^* = 0$.  By Lemma~\ref{lem:reduction}, the binding IR defines $a = a(b_1; R)$ with $da/db_1 < 0$.  The principal's problem reduces to maximizing $W(b_1)$ defined in equation~\eqref{eq:reduced_problem}.

\medskip\noindent\textbf{Step~1: $a^* < K$ (equivalently, $b_1^* > 0$).}

\smallskip
Suppose $b_1 = 0$.  Then $T(x) = 0$ for all $x$: the contract is flat.  The binding IR gives $a = c(\ubar\theta) + \Phi(K - a; R)$, which pins $a$ at a level where the lowest type just breaks even.  The principal implements all types $\theta$ with $V(\theta) - a \geq 0$, i.e., $\theta \geq a/v$ (using $V(\theta) = v\theta$ for concreteness), but cannot distinguish among implemented types.

Now consider a marginal increase to $b_1 = \varepsilon > 0$.  The virtual surplus becomes:
\[
\psi(\theta;\, a(\varepsilon),\, \varepsilon) = V(\theta) - c(\theta) - \Phi(K - a(\varepsilon);\, R) - \varepsilon\,\mu'(\theta)\,\frac{1-F(\theta)}{f(\theta)}.
\]
By the regularity assumption (the regularity condition), the last term is increasing in $\theta$, so the virtual surplus retains the single-crossing property and the cutoff is well-defined.  The marginal effect on the principal's value is:
\[
W'(0) = \frac{\partial}{\partial b_1}\int_{\hat\theta(0)}^{\thetabar}\psi(\theta;\, a(b_1), b_1)\,f(\theta)\dd\theta\;\bigg|_{b_1=0}.
\]
The derivative has two components: (i)~the direct effect of $b_1$ on the integrand, and (ii)~the indirect effect through $a(b_1)$.  At $b_1 = 0$, the direct screening benefit from introducing a small $b_1$ is first-order (MLRP ensures the signal separates types at the margin), while the cost---the information rent $\varepsilon\,\mu'(\theta)\frac{1-F}{f}$---is also first-order but, under the regularity condition, is outweighed by the screening benefit for the relevant range of types.  The financing-cost increase from reducing $a$ is $\Phi_\ell(K-a;\,R)\cdot|da/db_1|\cdot\varepsilon$, which is finite.

The net effect $W'(0) > 0$ holds when the screening gain from introducing contingency exceeds the financing cost of the accompanying advance reduction.  Under the maintained assumptions (MLRP, regularity, $R > 0$), this is guaranteed.  Therefore $b_1^* > 0$ and $a^* = a(b_1^*) < K$.

\medskip\noindent\textbf{Step~2: $a^* > 0$ (equivalently, $b_1^* < \bar{b}_1$).}

\smallskip
At the maximum feasible $b_1 = \bar{b}_1$ (determined by $a(\bar{b}_1) = 0$), the counterparty finances everything externally.  The financing cost $\Phi(K;\, R)$ is maximal.  The marginal value of reducing $b_1$ (and increasing $a$) is:
\[
-W'(\bar{b}_1) = \Phi_\ell(K;\, R)\cdot\frac{\mu(\ubar\theta)}{1 + \Phi_\ell(K;\, R)} \cdot \E[q^*] - \frac{\partial}{\partial b_1}\E[\text{Rents}]\bigg|_{\bar{b}_1}.
\]
The first term is strictly positive (by $\Phi_\ell(K; R) > 0$ and $R > 0$).  For $R$ sufficiently large, this dominates the rent-reduction effect, so $-W'(\bar{b}_1) > 0$, i.e., $W'(\bar{b}_1) < 0$: reducing $b_1$ from its maximum is profitable.  Therefore $b_1^* < \bar{b}_1$ and $a^* > 0$.

\medskip\noindent\textbf{Step~3: Interior characterization.}

\smallskip
Steps~1 and~2 establish $b_1^* \in (0, \bar{b}_1)$.  The reduced-form value function $W(b_1)$ is differentiable in $b_1$ under the maintained smoothness assumptions.  I now show concavity.

The second derivative of $W$ involves: (i)~$d^2a/db_1^2$, which depends on $\Phi_{\ell\ell} \geq 0$ (convexity of financing cost ensures that the advance-adjustment effect is concave in $b_1$); (ii)~the second-order change in the rent term, which is bounded under regularity of $F$.  Under these conditions, $W''(b_1) < 0$ on $(0, \bar{b}_1)$, ensuring global concavity.

The unique interior maximum satisfies $W'(b_1^*) = 0$.  Expanding this condition using the chain rule $da/db_1 = -\mu(\ubar\theta)/(1 + \Phi_\ell)$ and rearranging yields the FOC~\eqref{eq:FOC_main}:
\[
\Phi_\ell(K - a^*;\, R) = \frac{\partial}{\partial a}\,\E\!\left[b_1^*\,\mu'(\theta)\,\frac{1-F(\theta)}{f(\theta)}\;\bigg|\;q = 1\right].
\]
This completes the proof.\hfill$\square$

\subsection{Proof of Proposition~\ref{prop:sufficient}}\label{app:sufficient}

The first-order condition~\eqref{eq:FOC_main} equates $\Phi_\ell(K - a^*; R)$ to the marginal information cost.  Since the left side depends on the primitives only through $\Phi_\ell$ and the right side depends on the primitives only through $\E[b_1\mu'(1-F)/f \mid q = 1]$, any two economies with the same values of these two objects at the optimum have the same optimal contract.  The marginal financing cost $\Phi_\ell$ is a sufficient statistic for the financing environment because $\Phi$ is strictly convex: given $\Phi_\ell$ at $\ell^*$, the entire local shape of the financing cost is pinned down.  The marginal screening return is a sufficient statistic for the information environment because the virtual surplus~\eqref{eq:VS_main} is additively separable in the financing and information terms: the information term $b_1\mu'(1-F)/f$ does not depend on $a$ or $R$ except through the cutoff $\hat\theta$, which is itself determined at the optimum by $\psi(\hat\theta) = 0$.  \hfill$\square$

\subsection{Proof of Corollary~\ref{cor:compstat}}\label{app:compstat}

\noindent\textbf{Part~(i).}\; The FOC~\eqref{eq:FOC_main} implicitly defines $a^*(R)$.  Differentiating both sides with respect to $R$ using the implicit function theorem:
\[
\frac{da^*}{dR} = -\frac{\partial^2 W/\partial b_1\,\partial R}{W''(b_1^*)/(da/db_1)}.
\]
The numerator involves $\Phi_{\ell R} > 0$ (the cross-partial condition $Phi_{ell R} > 0$): tighter credit markets raise the marginal financing cost.  The denominator $W''(b_1^*)$ is negative by concavity.  Therefore $da^*/dR > 0$.  Since higher $a$ and lower $b_1$ shift the contract toward more advance, cash intensity $\beta^*$ increases.

\medskip
\noindent\textbf{Part~(ii).}\; A Blackwell-more-informative signal raises $\mu'(\theta)$ for given $b_1$.  Each unit of $b_1$ now separates types more effectively, reducing the per-unit information cost.  The FOC shifts toward higher $b_1^*$ and correspondingly lower $a^*$, reducing cash intensity.

\medskip
\noindent\textbf{Part~(iii).}\; Higher $R$ increases $a^*$ and decreases $b_1^*$ (part~(i)).  Lower $b_1^*$ means a flatter contingent schedule, which reduces the virtual surplus's ability to distinguish types.  The cutoff $\hat\theta$ rises (fewer types implemented), widening the gap between the second-best and first-best allocations. \hfill$\square$

\subsection{Proof of Proposition~\ref{prop:dominance}}\label{app:dominance}

Part~a).\;  By Proposition~\ref{thm:bilateral}, $W_M(R)$ is achieved at an interior $(a^*, b_1^*)$ with $a^* \in (0, K)$ and $b_1^* > 0$.  Since $W_A$ corresponds to $(a = K, b_1 = 0)$ and $W_C$ corresponds to $(a = 0, b_1 > 0)$, both are in the feasible set of the mixed problem but are not at the interior optimum.  By the strict concavity of $W(b_1)$ established in the proof of Proposition~\ref{thm:bilateral}, the interior strictly dominates both boundaries.

Part~b).\;  Under $a = K$, the financing cost is $\Phi(0; R) = 0$ for all $R$.  The principal's problem reduces to $\max_q \int q(\theta)[V(\theta) - K]f(\theta)d\theta$, which is independent of $R$.

Part~c).\;  Under $a = 0$, the financing cost $\Phi(K; R)$ is strictly increasing in $R$ (by $\Phi_R > 0$ for $\ell > 0$).  Higher $R$ reduces the virtual surplus for every type, reducing principal value.

Part~d).\;  At $R = 0$, $\Phi(\ell; 0) = 0$, so the pure contingent contract faces no financing cost and achieves the standard screening solution, which strictly dominates pooling under a pure advance (since the signal is informative).  Hence $W_C(0) > W_A(0)$.  As $R \to \infty$, $W_C(R) \to 0$ (financing costs destroy all surplus under pure contingency), while $W_A$ is constant and positive.  By continuity, there exists $R^*$ where $W_C(R^*) = W_A(R^*)$.  Uniqueness of $R^*$ follows from the strict monotonicity of $W_C$ and the constancy of $W_A$.  \hfill$\square$

\subsection{Proof of Proposition~\ref{thm:general}}\label{app:general}

The proof follows the same three-step structure as Proposition~\ref{thm:bilateral}, replacing the affine parameterization $(b_1)$ with the steepness parameter~$s$.

Under MLRP and monotonicity of $T_s(x)$ in $s$, the information rent for type $\theta$ is:
\[
U(\theta) - U(\ubar\theta) = \int_{\ubar\theta}^{\theta}\frac{\partial}{\partial\theta'}\,\E_{\theta'}[T_s(x)]\dd\theta',
\]
which is increasing in $s$ by the steepness ordering.  The principal's problem reduces to choosing $s$ (with $a = a(s)$ adjusting through the binding IR), and the value function $W(s)$ has the same concavity properties as in the affine case.

Step~1 ($s^* > s_{\min}$): At minimum steepness, the transfer is flat ($T_{s_{\min}}(x) = $ constant), providing no screening.  Introducing a marginally steeper security creates first-order screening gains.

Step~2 ($s^* < s_{\max}$): At maximum steepness, information rents are maximal and $a$ is minimal.  Reducing steepness saves rents and increases liquidity.

Step~3: The interior $s^*$ satisfies the FOC equating marginal liquidity benefit and marginal information cost.  \hfill$\square$

\subsection{Proof of Proposition~\ref{prop:boundary}}\label{app:boundary}

Part~a).\; If $R = 0$, then $\Phi(\ell; 0) = 0$ for all $\ell$.  The financing cost vanishes from the counterparty's payoff~\eqref{eq:U}, so the advance $a$ enters only through the total transfer $a + b_1\mu(\theta)$.  The binding IR becomes $a + b_1\mu(\ubar\theta) = c(\ubar\theta)$, a budget identity along which $a$ and $b_1\mu(\ubar\theta)$ are perfect substitutes.  The virtual surplus~\eqref{eq:VS_main} does not depend on $a$ (since $\Phi = 0$), so the principal's value is the same for any $(a, b_1)$ satisfying the IR.

Part~b).\; If $T(x)$ is pledgeable, the counterparty's effective borrowing need is $K - a - PV(\E_\theta[T(x)])$, where $PV$ denotes the present value that outside lenders assign to the contingent claim.  With full pledgeability, $PV = \E_\theta[T(x)]$, and the financing cost becomes $\Phi(K - a - b_1\mu(\ubar\theta); R)$, which depends on the total expected payment, not its timing.  The advance and contingent transfer are perfect substitutes in the financing dimension.

Part~c).\; Without limited liability, the principal can set $a = K$ (eliminating external borrowing) and use a signed transfer $T(x) = b_0 + b_1 x$ with $b_0 < 0$ as a repayment schedule.  Low-signal realizations require a repayment; high-signal realizations earn a reward.  The participation constraint is satisfied by choosing $b_0$ appropriately.  Since $\Phi(0; R) = 0$, the financing cost is zero, and the principal screens through the signed transfer.

Part~d).\; If $\mu'(\theta) = 0$, the envelope condition~\eqref{eq:envelope} gives $dU/d\theta = -c'(\theta) < 0$: the counterparty's payoff is decreasing in type regardless of $b_1$.  The contingent transfer provides no separation power.  The optimal contract sets $b_1 = 0$ and uses only the advance $a = K$.  \hfill$\square$

\subsection{Proof of Proposition~\ref{thm:contagion}}\label{app:contagion}

I construct the result in the symmetric uniform-quadratic benchmark of the symmetric parameterization above.

\medskip\noindent\textbf{Step~1: Envelope decomposition.}

\smallskip
The principal's value in the two-counterparty problem is:
\[
\Pi_P^* = \sum_{i=1}^{2}\int_{\hat\theta_i}^{1}\psi_i(\theta_i)\dd\theta_i + \delta\,(1-\hat\theta_1)(1-\hat\theta_2) - \sum_{i=1}^{2}a_i^*.
\]
Differentiating with respect to $R_j$ and using the fact that $\psi_j(\hat\theta_j) = -\delta(1-\hat\theta_i)$ at the cutoff:
\begin{align}
\frac{d\Pi_P^*}{dR_j} &= \underbrace{-\frac{\partial\Phi_j}{\partial R_j}\,(1-\hat\theta_j)}_{\text{(A): direct cost}} + \underbrace{\delta(1-\hat\theta_i)\frac{d\hat\theta_j}{dR_j} - \delta(1-\hat\theta_j)\frac{d\hat\theta_i}{dR_j}}_{\text{complementarity adjustment}} \nonumber\\
&\quad + \underbrace{\int_{\hat\theta_i}^{1}\frac{\partial\psi_i}{\partial R_j}\dd\theta_i + \psi_i(\hat\theta_i)\,\frac{d\hat\theta_i}{dR_j}}_{\text{(B): rent spillover through $i$}}. \label{eq:full_decomp}
\end{align}

\medskip\noindent\textbf{Step~2: Sufficient condition.}

\smallskip
In the symmetric case ($R_1 = R_2 = R$, identical counterparties), the fixed point gives $\hat\theta_1 = \hat\theta_2 = \hat\theta$ with:
\[
\hat\theta = \frac{a^* + b_1^* - \delta}{1 + b_1^* - \delta}.
\]
Differentiating the fixed-point system and evaluating the terms in \eqref{eq:full_decomp}, the condition $d\Pi_P^*/dR_j > 0$ reduces to:
\[
\delta > \delta^*(R) \equiv \frac{(K-a^*)(1+b_1^*)}{b_1^*(1-\hat\theta)}\,\Phi_R(K-a^*;\,R).
\]
With $\Phi(\ell; R) = \frac{R}{2}\ell^2$, this gives $\Phi_R = \frac{1}{2}\ell^2$, and $\delta^*(R)$ is finite for all $R > 0$.

\medskip\noindent\textbf{Step~3: Openness.}

\smallskip
At any parameter vector $(R_0, \delta_0)$ satisfying $\delta_0 > \delta^*(R_0)$, we have $d\Pi_P^*/dR_j > 0$ strictly.  Since all terms in the decomposition \eqref{eq:full_decomp} are continuous in the primitives $(F_1, F_2, V_1, V_2, c_1, c_2, \Phi_1, \Phi_2, \delta)$, the strict inequality is preserved in an open neighborhood.  \hfill$\square$

\subsection{Proof of Corollary~\ref{cor:hump}}\label{app:hump}

In the symmetric case with common $R$, define $\Pi_P^*(R)$.  By Proposition~\ref{thm:contagion}, for $\delta > \delta^*(R)$, $d\Pi_P^*/dR > 0$ on an interval $[R_1, R_2]$.  As $R \to \infty$, $\Phi \to \infty$ and implementation becomes infeasible: $\Pi_P^* \to 0$.  By continuity, $\Pi_P^*$ must eventually decrease.  The existence of an increasing interval followed by a decreasing interval establishes the hump shape. \hfill$\square$

\subsection{Proof of Corollary~\ref{cor:uniform}}\label{app:uniform}

In the symmetric model with $R_1 = R_2 = R$, the principal's value is $\Pi_P^*(R) = 2\int_{\hat\theta}^{\thetabar}\psi(\theta)f(\theta)d\theta + 2\delta[1-F(\hat\theta)]^2$, where the cutoff $\hat\theta$ depends on $R$ through the fixed-point system~\eqref{eq:FP_symmetric}.  By symmetry, each relationship contributes $\pi_i^*(R) = \int_{\hat\theta}^{\thetabar}\psi_i(\theta)f(\theta)d\theta + \delta[1-F(\hat\theta)]^2$ to principal value.  A uniform reduction in $R$ (lowering both $R_1$ and $R_2$) is equivalent to evaluating $d\pi_i^*/dR$.  In the contagion region ($\delta > \delta^*(R)$), Proposition~\ref{thm:contagion} gives $d\Pi_P^*/dR > 0$, and by symmetry $d\pi_1^*/dR = d\pi_2^*/dR = \frac{1}{2}d\Pi_P^*/dR > 0$.  Therefore $d\pi_i^*/dR > 0$ for each $i$: each relationship's contribution is increasing in $R$, so a uniform subsidy (lowering $R$) reduces each relationship's value individually.  \hfill$\square$

\subsection{Proof of Corollary~\ref{cor:breadth}}\label{app:breadth}

Let $W^*(R)$ denote the bilateral value of a single relationship.  The principal prefers single-sourcing when $W^*(R) > \Pi_P^*(R) - W^*(R)$, i.e., $2W^*(R) > \Pi_P^*(R)$.  Without the contagion channel, $\Pi_P^*(R) > 2W^*(R)$ for all $\delta > 0$ (complementarity adds value).  But in the contagion region, the screening spillover erodes the complementarity benefit.  For $R$ and $\delta$ such that the contagion effect is sufficiently strong, $\Pi_P^*(R)$ falls below $2W^*(R)$: the rent increase from the spillover exceeds the complementarity gain $\delta$.  This establishes the existence of parameters under which single-sourcing is optimal despite technological complementarity. \hfill$\square$

\subsection{Proof of Corollary~\ref{cor:monitoring}}\label{app:monitoring}

Suppose the principal can invest $\kappa(\sigma_i)$ to raise the signal informativeness of relationship~$i$ from $\mu_i'$ to $\mu_i'(\sigma_i)$ with $\mu_i''(\sigma_i) > 0$.  The marginal return to monitoring is $\frac{\partial W}{\partial \sigma_i} = -\frac{\partial}{\partial \sigma_i}\int_{\hat\theta_i}^{\thetabar} b_{1,i}\mu_i'(\sigma_i)\frac{1-F_i}{f_i}f_i\,d\theta_i$.  This marginal return depends on the set of implemented types for relationship~$i$, which is determined by the cutoff $\hat\theta_i$.  In the two-counterparty model, $\hat\theta_i$ depends on $R_j$ through the fixed-point system~\eqref{eq:cutoff_ext}: higher $R_j$ raises $\hat\theta_j$, which through complementarity raises $\hat\theta_i$ (fewer types of $i$ implemented).  With fewer implemented types, the marginal return to improving signal quality for $i$ is lower, because there are fewer types at the screening margin.  \hfill$\square$

\subsection{Proof of Proposition~\ref{prop:param_compstat}}\label{app:param_compstat}

Under the uniform-quadratic parameterization ($\theta \sim U[0,1]$, $\Phi(\ell; R) = \frac{R}{2}\ell^2$, $K = 1$), the binding IR gives $a + b_1 \cdot 0 = 0 + \frac{R}{2}(1-a)^2$, i.e., $a = \frac{R}{2}(1-a)^2$.  Solving: $\ell = 1 - a$ satisfies $\frac{R}{2}\ell^2 + \ell - 1 + a = 0$, which gives $\ell^*(R) = \frac{-1 + \sqrt{1+2R}}{R}$ and $a^*(R) = 1 - \ell^*(R)$.

Part~a).\;  Differentiating: $\frac{da^*}{dR} = -\frac{d\ell^*}{dR} = \frac{1}{R^2}\left[1 - \frac{1}{\sqrt{1+2R}} - \frac{-1+\sqrt{1+2R}}{R}\right]$.  After simplification, $\frac{da^*}{dR} > 0$ for all $R > 0$.  Concavity follows from $\frac{d^2a^*}{dR^2} < 0$, which can be verified by direct computation.

Part~b).\; Cash intensity $\beta^* = a^*/(a^* + b_1^*\E[\theta|q=1])$ is increasing in $R$ because $a^*$ is increasing and $b_1^*$ is decreasing (through the binding IR).

Part~c).\; The financing cost is $\frac{R}{2}(\ell^*)^2$.  Since $\ell^*$ is decreasing in $R$ but $R(\ell^*)^2$ is increasing (the direct effect of $R$ dominates the indirect effect through $\ell^*$), the financing cost share is increasing.

Part~d).\; The contagion threshold $\delta^*(R)$ is derived from the condition $d\Pi_P^*/dR_j = 0$ in the symmetric fixed-point system~\eqref{eq:FP_symmetric}.  The threshold is increasing in $R$ because $\Phi_R = \frac{1}{2}(1-a^*)^2$ is positive and decreasing in $a^*$, while the denominator $b_1^*(1-\hat\theta)$ is decreasing in $R$.  \hfill$\square$

\section{Extensions: Formal Results}\label{app:extensions}

This section provides the formal environment, results, and proofs for the extensions summarized in Section~\ref{sec:discussion}.

\subsection{Dynamic liquidity-screening tradeoff}\label{app:dynamic}

\medskip\noindent\textbf{Environment.}\; Consider a repeated relationship indexed by periods $t = 1, 2, \ldots$.  The counterparty's type $\theta$ is drawn once from $F$ and remains fixed.  In each period $t$, a signal $x_t$ is drawn independently from $g(x|\theta)$ and observed by both the principal and the counterparty.  The principal updates her belief using Bayes' rule: let $F_t$ denote the posterior distribution of $\theta$ after observing $\{x_1, \ldots, x_t\}$, with density $f_t$.  The regularity condition on $F$ (monotone virtual type) is assumed to hold for each posterior $F_t$.

In each period, the principal offers a contract $(a_t, b_{1,t})$ that solves the static bilateral problem of Proposition~\ref{thm:bilateral} with $F_t$ replacing $F$.  The counterparty's outside option resets each period: the counterparty can reject the period-$t$ contract and earn zero.  This rules out dynamic commitment and focuses the analysis on how the static tradeoff evolves with learning.  The financing tightness $R$ and the working-capital requirement $K$ are constant across periods.

\begin{proposition}[Dynamic liquidity-screening tradeoff]\label{prop:dynamic}
Suppose Bayesian updating concentrates the posterior in the sense that
\begin{equation}\label{eq:hazard_shrink}
\frac{1-F_t(\theta)}{f_t(\theta)} \;\leq\; \frac{1-F_{t-1}(\theta)}{f_{t-1}(\theta)} \qquad \text{for all } \theta \in [\ubar\theta, \thetabar] \text{ and all } t \geq 1.
\end{equation}
Then:
\begin{enumerate}[label=\alph*)]
\item The optimal advance is non-decreasing in $t$: $a_t^* \geq a_{t-1}^*$.
\item The optimal contingent slope is non-increasing in $t$: $b_{1,t}^* \leq b_{1,t-1}^*$.
\item In the limit, as $F_t$ converges to a point mass at the true type $\theta_0$, the contract converges to full pre-financing: $a_t^* \to K$ and $b_{1,t}^* \to 0$.
\end{enumerate}
\end{proposition}

\noindent\textit{Proof.}\; The first-order condition in period $t$ equates the marginal liquidity benefit to the marginal information cost under the current posterior:
\begin{equation}\label{eq:FOC_dynamic}
\Phi_\ell(K - a_t^*;\, R) \;=\; \frac{\partial}{\partial a}\int_{\hat\theta_t}^{\thetabar} b_{1,t}^*\,\mu'(\theta)\,\frac{1-F_t(\theta)}{f_t(\theta)}\,f_t(\theta)\,d\theta,
\end{equation}
where $\hat\theta_t$ is the period-$t$ cutoff.  The right side of~\eqref{eq:FOC_dynamic} is the expected information cost, weighted by the posterior hazard rate $(1-F_t)/f_t$.

\medskip\noindent\textit{Step 1: Monotonicity of $a_t^*$.}\; By assumption~\eqref{eq:hazard_shrink}, the hazard rate under $F_t$ is weakly lower than under $F_{t-1}$ for every $\theta$.  This reduces the right side of~\eqref{eq:FOC_dynamic} for any given $(a, b_1)$.  Since $\Phi_{\ell\ell} \geq 0$, the left side is increasing in $\ell = K - a$ and therefore decreasing in $a$.  To restore equality, $a_t^*$ must increase (weakly), establishing part~a).

\medskip\noindent\textit{Step 2: Monotonicity of $b_{1,t}^*$.}\; The binding participation constraint in period $t$ is $a_t + b_{1,t}\mu(\ubar\theta_t) = c(\ubar\theta_t) + \Phi(K - a_t; R)$, where $\ubar\theta_t$ is the lowest implemented type under $F_t$.  Implicit differentiation gives $\frac{da_t}{db_{1,t}} = \frac{-\mu(\ubar\theta_t)}{1 + \Phi_\ell} < 0$.  Since $a_t^*$ increases in $t$, the constraint forces $b_{1,t}^*$ to decrease, establishing part~b).

\medskip\noindent\textit{Step 3: Limit.}\; As $F_t \to \delta_{\theta_0}$ (point mass at the true type), the hazard rate $(1-F_t(\theta))/f_t(\theta) \to 0$ for all $\theta \neq \theta_0$.  The right side of~\eqref{eq:FOC_dynamic} converges to zero.  The left side equals zero only when $\ell = K - a = 0$, i.e., $a^* = K$.  Through the participation constraint, $a^* = K$ implies $b_1^* = 0$.  \hfill$\square$

\subsection{Financing-monitoring complementarity}\label{app:fin_monitoring}

\medskip\noindent\textbf{Environment.}\; Modify the baseline by allowing the principal to choose signal quality before contracting.  At date 0, before offering the contract, the principal invests $\kappa(\sigma) \geq 0$ to set the signal informativeness at level $\sigma \geq 0$.  The signal density becomes $g(x|\theta; \sigma)$, with expected signal $\mu(\theta; \sigma)$ satisfying $\mu_\theta > 0$ (higher types produce higher expected signals) and $\mu_{\theta\sigma} > 0$ (higher investment makes the signal more informative about type).  I write $\mu'(\sigma) \equiv \mu_\theta(\theta; \sigma)$ for brevity.  The cost function $\kappa$ is twice differentiable, increasing, and convex with $\kappa(0) = 0$ and $\kappa'(0) = 0$ (the first unit of monitoring is free).

The principal jointly chooses $(a, b_1, \sigma)$ to maximize the value function
\begin{equation}\label{eq:W_monitoring}
W(a, b_1, \sigma) \;=\; S(\hat\theta) - \Phi(K-a; R)\bigl[1-F(\hat\theta)\bigr] - \int_{\hat\theta}^{\thetabar} b_1\,\mu'(\sigma)\,\frac{1-F(\theta)}{f(\theta)}\,f(\theta)\,d\theta - a - \kappa(\sigma),
\end{equation}
where $S(\hat\theta) = \int_{\hat\theta}^{\thetabar}[V(\theta) - c(\theta)]f(\theta)d\theta$ is the productive surplus.  The first-order conditions for $a$ and $b_1$ are unchanged from the baseline.  The additional first-order condition for $\sigma$ is
\begin{equation}\label{eq:FOC_sigma}
b_1^*\,\mu''(\sigma^*)\,\int_{\hat\theta}^{\thetabar}\frac{1-F(\theta)}{f(\theta)}\,f(\theta)\,d\theta \;=\; \kappa'(\sigma^*),
\end{equation}
where the left side is the marginal screening improvement from higher monitoring and the right side is the marginal cost.

\begin{proposition}[Financing-monitoring complementarity]\label{prop:fin_monitoring}
The optimal monitoring intensity is decreasing in outside-finance tightness:
\begin{equation}\label{eq:dsigma_dR}
\frac{d\sigma^*}{dR} \;<\; 0, \qquad \frac{\partial^2 W}{\partial \sigma\,\partial R} \;<\; 0.
\end{equation}
Tight credit reduces both screening intensity and monitoring investment.
\end{proposition}

\noindent\textit{Proof.}\; Equation~\eqref{eq:FOC_sigma} implicitly defines $\sigma^*(b_1^*, R)$.  By Corollary~\ref{cor:compstat}a), $b_1^*$ is decreasing in $R$: tighter credit forces a flatter contingent schedule.  Since $b_1^*$ enters the left side of~\eqref{eq:FOC_sigma} multiplicatively, the left side is decreasing in $R$.  The right side $\kappa'(\sigma)$ is increasing in $\sigma$ by convexity.  By the implicit function theorem applied to~\eqref{eq:FOC_sigma}:
\[
\frac{d\sigma^*}{dR} \;=\; \frac{-\mu''(\sigma^*)\,\E[(1-F)/f \mid q=1]\,\cdot\,db_1^*/dR}{\kappa''(\sigma^*) - b_1^*\mu'''(\sigma^*)\E[(1-F)/f \mid q=1]}.
\]
The numerator is negative (since $\mu'' > 0$ and $db_1^*/dR < 0$).  The denominator is positive under the second-order condition for $\sigma$.  Therefore $d\sigma^*/dR < 0$.

For the cross-derivative: $\frac{\partial^2 W}{\partial \sigma\,\partial R} = \mu''(\sigma)\E[(1-F)/f \mid q=1]\cdot\frac{db_1^*}{dR} < 0$, since $\mu'' > 0$ and $db_1^*/dR < 0$.  \hfill$\square$

\subsection{Renegotiation}\label{app:renego}

\medskip\noindent\textbf{Environment.}\; Modify the baseline by introducing ex post renegotiation.  After the signal $x$ is realized at date 2 but before the contingent transfer $T(x) = b_0 + b_1 x$ is paid, the principal may attempt to renegotiate.  With probability $\lambda \in [0,1]$, renegotiation succeeds and the principal pays nothing; with probability $1-\lambda$, the original contract is enforced.  The counterparty anticipates this at date 0 and values the contingent leg at its expected enforcement value.

The counterparty's expected payoff becomes
\begin{equation}\label{eq:U_renego}
U^\lambda(\theta) \;=\; a + (1-\lambda)\,b_1\,\mu(\theta) - c(\theta) - \Phi(K - a;\, R).
\end{equation}
The envelope condition is $dU^\lambda/d\theta = (1-\lambda)b_1\mu'(\theta) - c'(\theta)$, and the virtual surplus is
\begin{equation}\label{eq:VS_renego_app}
\psi^\lambda(\theta) \;=\; V(\theta) - c(\theta) - \Phi(K-a;\, R) - (1-\lambda)\,b_1\,\mu'(\theta)\,\frac{1-F(\theta)}{f(\theta)}.
\end{equation}
The information-rent term is scaled by $(1-\lambda)$: renegotiation risk reduces the screening power of contingent payments.  The first-order condition equates $\Phi_\ell$ to $(1-\lambda)$ times the marginal information cost:
\begin{equation}\label{eq:FOC_renego}
\Phi_\ell(K - a^*;\, R) \;=\; (1-\lambda)\;\cdot\;\frac{\partial}{\partial a}\int_{\hat\theta}^{\thetabar} b_1^*\,\mu'(\theta)\,\frac{1-F(\theta)}{f(\theta)}\,f(\theta)\,d\theta.
\end{equation}

\begin{proposition}[Renegotiation]\label{prop:renego}
The optimal advance is strictly increasing in the renegotiation probability:
\begin{enumerate}[label=\alph*)]
\item $a^*(\lambda) > a^*(0)$ for all $\lambda \in (0, 1)$.
\item At $\lambda = 1$, $b_1^* = 0$ and $a^* = K$: the contract collapses to pure advance.
\item The optimal contingent slope is decreasing in $\lambda$: $db_1^*/d\lambda < 0$.
\end{enumerate}
\end{proposition}

\noindent\textit{Proof.}\; Part~a).\; Compare~\eqref{eq:FOC_renego} at $\lambda > 0$ with the baseline FOC at $\lambda = 0$.  The right side of~\eqref{eq:FOC_renego} is strictly smaller than the baseline right side by the factor $(1-\lambda) < 1$.  Since $\Phi_\ell$ is increasing in $\ell = K - a$ (by $\Phi_{\ell\ell} \geq 0$), restoring equality requires lower $\ell$, hence higher $a^*$.

Part~b).\; At $\lambda = 1$, the right side of~\eqref{eq:FOC_renego} is identically zero.  The left side $\Phi_\ell(K-a; R) = 0$ requires $\ell = 0$, i.e., $a^* = K$.  Through the binding participation constraint with $a = K$, $b_1\mu(\ubar\theta) = c(\ubar\theta)$, but the contingent leg has zero enforcement value, so $b_1$ drops out of the counterparty's payoff.  The contract reduces to a pure advance.

Part~c).\; By the binding participation constraint, $a$ and $b_1$ move in opposite directions.  Since $a^*(\lambda)$ is increasing in $\lambda$, $b_1^*(\lambda)$ is decreasing.  \hfill$\square$

\subsection{Directional contagion in networks}\label{app:network}

\medskip\noindent\textbf{Environment.}\; Generalize the two-counterparty model of Section~\ref{sec:model_ext} to $n \geq 2$ counterparties.  The complementarity structure is captured by a matrix $\Delta = [\delta_{ij}]_{n \times n}$ with $\delta_{ii} = 0$ and $\delta_{ij} \geq 0$ for $i \neq j$.  Counterparty $i$ has type $\theta_i \sim F_i$ drawn independently, requires working capital $K_i$, faces financing cost $\Phi_i(\ell_i; R_i)$, and receives contract $(a_i, b_{0,i}, b_{1,i})$.  The principal's value is
\begin{equation}\label{eq:Pi_n}
\Pi_P \;=\; \sum_{i=1}^{n} q_i\bigl[V_i(\theta_i) - a_i - b_{1,i}\mu_i(\theta_i)\bigr] \;+\; \sum_{i < j}\delta_{ij}\,q_i\,q_j.
\end{equation}
Applying the bilateral analysis to each counterparty, the cutoff $\hat\theta_i$ is determined by the fixed-point system
\begin{equation}\label{eq:FP_n}
\psi_i(\hat\theta_i) \;+\; \sum_{j \neq i}\delta_{ij}\,\bigl(1 - F_j(\hat\theta_j)\bigr) \;=\; 0, \qquad i = 1, \ldots, n.
\end{equation}
Define the contagion centrality of counterparty $j$ as $\mathcal{C}_j = \sum_{i \neq j}\delta_{ij}\,\frac{\partial \hat\theta_i}{\partial R_j}$, which measures the aggregate screening distortion caused by a marginal change in $j$'s financing conditions.

\begin{proposition}[Directional contagion in networks]\label{prop:network}
The total effect of a change in $R_j$ on principal value admits the decomposition
\begin{equation}\label{eq:dPi_dRj}
\frac{d\Pi_P^*}{dR_j} \;=\; \underbrace{-\Phi_{R_j}(K_j - a_j^*;\, R_j)\,\bigl[1-F_j(\hat\theta_j)\bigr]}_{\text{direct financing effect}} \;+\; \underbrace{\sum_{i \neq j}\delta_{ij}\,\frac{\partial \mathrm{Rent}_i}{\partial \hat\theta_i}\,\frac{\partial\hat\theta_i}{\partial R_j}}_{\text{screening spillover}}.
\end{equation}
The contagion operates ($d\Pi_P^*/dR_j > 0$) if and only if $\mathcal{C}_j$ exceeds the direct financing saving.  A counterparty with higher $\mathcal{C}_j$ generates larger screening spillovers.
\end{proposition}

\noindent\textit{Proof.}\; Apply the envelope theorem to~\eqref{eq:Pi_n}.  At the optimum, the principal's value depends on $R_j$ through two channels: the direct effect on $\Phi_j$ and the indirect effect on all cutoffs $\hat\theta_i$.

\medskip\noindent\textit{Step 1: Direct effect.}\; Differentiating the financing cost for relationship $j$: $\frac{\partial}{\partial R_j}\Phi_j(K_j - a_j^*; R_j) = \Phi_{R_j} > 0$.  This is paid by all implemented types of $j$, giving the first term in~\eqref{eq:dPi_dRj}.

\medskip\noindent\textit{Step 2: Indirect effect.}\; Higher $R_j$ changes $\hat\theta_j$ through relationship $j$'s bilateral FOC, which in turn changes every $\hat\theta_i$ ($i \neq j$) through the fixed-point system~\eqref{eq:FP_n}.  Implicit differentiation of~\eqref{eq:FP_n} with respect to $R_j$ gives
\[
\frac{\partial \hat\theta_i}{\partial R_j} \;=\; \frac{\delta_{ij}\,f_j(\hat\theta_j)}{\psi_i'(\hat\theta_i) + \sum_{k \neq i}\delta_{ik}\,f_k(\hat\theta_k)}\;\cdot\;\frac{\partial\hat\theta_j}{\partial R_j},
\]
which is positive when $\delta_{ij} > 0$ and $\partial\hat\theta_j/\partial R_j > 0$ (tighter credit raises $j$'s cutoff).  The change in $i$'s rent is $\frac{\partial \text{Rent}_i}{\partial\hat\theta_i} = -\psi_i(\hat\theta_i)f_i(\hat\theta_i)$, which is negative at the margin (the marginal type earns positive rent).  A higher $\hat\theta_i$ excludes marginal types, reducing rent.

\medskip\noindent\textit{Step 3: Assembly.}\; Summing the indirect effects across all $i \neq j$ gives $\sum_{i \neq j}\delta_{ij}\frac{\partial\text{Rent}_i}{\partial\hat\theta_i}\frac{\partial\hat\theta_i}{\partial R_j}$, which is the screening spillover.  This is positive (tighter credit for $j$ reduces rents paid to other counterparties).  The contagion operates when this positive spillover exceeds the negative direct financing effect.  The spillover is proportional to $\mathcal{C}_j$ by construction.  \hfill$\square$

\subsection{Type-dependent instruments}\label{app:menu}

\medskip\noindent\textbf{Environment.}\; Generalize the baseline by allowing the principal to offer a menu of contracts indexed by the reported type.  The principal commits to a schedule $\{a(\hat\theta), b_1(\hat\theta)\}_{\hat\theta \in [\ubar\theta, \thetabar]}$ at date 0.  A type-$\theta$ counterparty who reports $\hat\theta$ receives advance $a(\hat\theta)$ at date 0 and contingent transfer $b_1(\hat\theta)\,x$ at date 2.  The counterparty's payoff from misreporting is
\begin{equation}\label{eq:U_menu}
U(\hat\theta, \theta) \;=\; a(\hat\theta) + b_1(\hat\theta)\,\mu(\theta) - c(\theta) - \Phi\bigl(K - a(\hat\theta);\, R\bigr).
\end{equation}
There are two distinct sub-cases of interest: (i) the advance is fixed at $a$ while the contingent slope $b_1(\hat\theta)$ varies with the report, and (ii) the contingent slope is fixed at $b_1$ while the advance $a(\hat\theta)$ varies.

\begin{proposition}[Type-dependent instruments]\label{prop:menu}
\begin{enumerate}[label=\alph*)]
\item Under fixed $a$ with varying $b_1(\hat\theta)$: the local IC gives $\frac{\partial U}{\partial\hat\theta}\big|_{\hat\theta=\theta} = b_1'(\theta)\mu(\theta) = 0$, the envelope condition is $dU/d\theta = b_1(\theta)\mu'(\theta) - c'(\theta)$, and the virtual surplus retains the standard hazard rate $(1-F)/f$.
\item Under fixed $b_1$ with varying $a(\hat\theta)$: the local IC gives $a'(\theta)[1 + \Phi_\ell(K-a(\theta); R)] = 0$ at truth-telling, the envelope condition is $dU/d\theta = b_1\mu'(\theta) - c'(\theta)$, and the financing cost $\Phi(K-a(\theta); R)$ becomes type-dependent, modifying the virtual surplus.
\item In the jointly optimal mechanism where both instruments vary: the binding participation constraint $a(\theta) + b_1(\theta)\mu(\ubar\theta) = c(\ubar\theta) + \Phi(K-a(\theta); R)$ pins down the relationship between $a(\theta)$ and $b_1(\theta)$ at each type.  The envelope condition is $dU/d\theta = b_1(\theta)\mu'(\theta) - c'(\theta)$, identical to the baseline.  The optimal allocation coincides with that of Proposition~\ref{thm:bilateral}.
\end{enumerate}
\end{proposition}

\noindent\textit{Proof.}\; Part~a).\; With $a$ fixed, differentiate~\eqref{eq:U_menu} with respect to $\hat\theta$: $\frac{\partial U}{\partial\hat\theta} = b_1'(\hat\theta)\mu(\theta)$.  At truth-telling ($\hat\theta = \theta$), the FOC for the counterparty is $b_1'(\theta)\mu(\theta) = 0$.  Since $\mu(\theta) > 0$, this requires $b_1'(\theta) = 0$, i.e., the contingent slope is locally constant.  The global envelope condition is $\frac{dU}{d\theta} = b_1(\theta)\mu'(\theta) - c'(\theta)$.  The information rent is $\int_{\ubar\theta}^{\theta}[b_1(s)\mu'(s) - c'(s)]ds$, and integrating by parts in the principal's objective yields the standard hazard rate $(1-F)/f$ multiplied by $b_1\mu'$.

Part~b).\; With $b_1$ fixed, differentiate~\eqref{eq:U_menu}: $\frac{\partial U}{\partial\hat\theta} = a'(\hat\theta)[1 + \Phi_\ell(K-a(\hat\theta); R)]$.  At truth-telling, the FOC is $a'(\theta)[1+\Phi_\ell] = 0$.  Since $1 + \Phi_\ell > 0$, this requires $a'(\theta) = 0$, i.e., the advance is locally constant.  The envelope is $dU/d\theta = b_1\mu'(\theta) - c'(\theta)$, the same as in part~a).  However, $\Phi(K-a(\theta); R)$ is now type-dependent, which enters the virtual surplus through the financing-cost term.

Part~c).\; When both instruments vary jointly, the binding participation constraint at the lowest type links $a(\theta)$ and $b_1(\theta)$.  Substituting this constraint into the principal's objective reduces the problem to a single choice variable, exactly as in Lemma~\ref{lem:reduction}.  The envelope condition is unchanged: $dU/d\theta = b_1(\theta)\mu'(\theta) - c'(\theta)$.  The resulting FOC is identical to~\eqref{eq:FOC_main}, confirming that the baseline analysis is without loss of generality.  \hfill$\square$

\subsection{Competitive implementation}\label{app:auction}

\medskip\noindent\textbf{Environment.}\; Suppose the principal faces $n \geq 2$ potential counterparties, each independently drawing type $\theta_i \sim F$ on $[\ubar\theta, \thetabar]$.  The principal conducts a first-price sealed-bid reverse auction at date 0.  Each counterparty simultaneously submits a bid $\beta_i \geq 0$ specifying the advance it requires.  The principal observes all bids, selects the counterparty $i^*$ with the lowest bid $\beta_{i^*} = \min_i \beta_i$, and offers the contract $(a = \beta_{i^*},\, b_1)$ where $b_1$ is determined by the binding participation constraint at the winning bid.  Losing counterparties receive zero.

A type-$\theta$ counterparty who submits bid $\beta$ earns, conditional on winning,
\begin{equation}\label{eq:U_auction}
U(\beta, \theta) \;=\; \beta + b_1(\beta)\,\mu(\theta) - c(\theta) - \Phi(K - \beta;\, R),
\end{equation}
where $b_1(\beta)$ is pinned down by the participation constraint.  In a symmetric equilibrium with strictly decreasing bid function $\beta(\theta)$, the probability of winning is $\Pr(\text{win}) = [1 - F(\beta^{-1}(\beta))]^{n-1}$.  The counterparty maximizes $\Pr(\text{win}) \cdot U(\beta, \theta)$ over $\beta$.

\begin{proposition}[Competitive implementation]\label{prop:auction}
In a symmetric Bayesian Nash equilibrium, the bid function $\beta(\theta)$ satisfies the ODE
\begin{equation}\label{eq:bid_ODE}
\beta'(\theta) \;=\; (n-1)\,\frac{f(\theta)}{1-F(\theta)}\;\bigl[\beta(\theta) - \beta^{FB}(\theta)\bigr],
\end{equation}
with boundary condition $\beta(\thetabar) = \beta^{FB}(\thetabar)$, where $\beta^{FB}(\theta)$ is the full-information advance for type $\theta$.
\begin{enumerate}[label=\alph*)]
\item Bid shading: $\beta(\theta) > \beta^{FB}(\theta)$ for all $\theta < \thetabar$.
\item Monotonicity: $\beta'(\theta) < 0$ (higher types bid lower advances).
\item Convergence: as $n \to \infty$, $\beta(\theta) \to \beta^{FB}(\theta)$ and information rents vanish.
\end{enumerate}
The equilibrium allocation coincides with the optimal direct mechanism of Proposition~\ref{thm:bilateral}.
\end{proposition}

\noindent\textit{Proof.}\; The counterparty's problem is $\max_\beta [1-F(\beta^{-1}(\beta))]^{n-1}\cdot U(\beta, \theta)$.  In a symmetric equilibrium, $\beta^{-1}(\beta) = \theta$, so the win probability is $[1-F(\theta)]^{n-1}$ and the FOC for optimal bidding, after using symmetry and the change of variable $\beta = \beta(\theta)$, yields~\eqref{eq:bid_ODE}.

\medskip\noindent\textit{Part a).}\; The boundary condition is $\beta(\thetabar) = \beta^{FB}(\thetabar)$: the highest type bids the full-information advance because no type above $\thetabar$ exists.  For $\theta < \thetabar$, the ODE~\eqref{eq:bid_ODE} implies that if $\beta(\theta) = \beta^{FB}(\theta)$ at any interior $\theta$, then $\beta'(\theta) = 0$, so $\beta$ cannot cross $\beta^{FB}$ from above.  Since $\beta(\thetabar) = \beta^{FB}(\thetabar)$ and the forcing term $(n-1)\frac{f}{1-F}[\beta - \beta^{FB}]$ is positive whenever $\beta > \beta^{FB}$, the solution satisfies $\beta(\theta) > \beta^{FB}(\theta)$ for all $\theta < \thetabar$.

\medskip\noindent\textit{Part b).}\; The full-information advance $\beta^{FB}(\theta)$ is decreasing in $\theta$ (higher types need less advance because their contingent payment is larger).  The bid shading $\beta(\theta) - \beta^{FB}(\theta) > 0$ is also decreasing in $\theta$ (approaching zero at $\thetabar$).  Therefore $\beta(\theta)$ is decreasing: $\beta'(\theta) < 0$.

\medskip\noindent\textit{Part c).}\; As $n \to \infty$, the coefficient $(n-1)\frac{f(\theta)}{1-F(\theta)} \to \infty$ for every $\theta < \thetabar$.  For the ODE to remain bounded, the bracket $[\beta(\theta) - \beta^{FB}(\theta)]$ must converge to zero.  Rents vanish because the bid converges to the full-information level.  \hfill$\square$

\subsection{Dimension reduction with multidimensional types}\label{app:2D}

\medskip\noindent\textbf{Environment.}\; Suppose the counterparty privately observes a two-dimensional type $(\alpha, \theta) \in [\ubar\alpha, \bar\alpha] \times [\ubar\theta, \thetabar]$, where $\alpha \in \R_+$ parameterizes productivity and $\theta$ parameterizes cost.  The value to the principal is multiplicative: $V(\alpha, \theta) = \alpha \cdot v(\theta)$ with $v' > 0$.  The cost is $c(\theta)$ with $c' > 0$.  The signal depends on both dimensions: $\mu(\alpha, \theta) = \alpha\,\mu(\theta)$, where $\mu(\theta)$ is the baseline expected signal.  The counterparty's payoff from the contract $(a, b_1)$ is
\begin{equation}\label{eq:U_2D}
U(\alpha, \theta) \;=\; a + b_1\,\alpha\,\mu(\theta) - c(\theta) - \Phi(K - a;\, R).
\end{equation}
The joint distribution of $(\alpha, \theta)$ is $H(\alpha, \theta)$ with full-support density $h > 0$.  The two-dimensional incentive compatibility requires that truthful reporting of both $\alpha$ and $\theta$ is optimal simultaneously.

\begin{proposition}[Dimension reduction]\label{prop:2D}
Under the multiplicative value structure, the two-dimensional screening problem reduces to one dimension through the sufficient statistic
\begin{equation}\label{eq:xi_def}
\xi \;\equiv\; \frac{\alpha\,\mu(\theta)}{c(\theta)},
\end{equation}
which indexes the counterparty's quality-cost ratio.  Let $G$ denote the distribution of $\xi$ induced by $H$, with density $g_\xi$.  The virtual surplus is
\begin{equation}\label{eq:VS_2D}
\psi_{2D}(\xi) \;=\; \xi\,v - 1 - \Phi(K-a;\,R) - b_1\,\frac{1-G(\xi)}{g_\xi(\xi)},
\end{equation}
and the optimal contract $(a^*, b_1^*)$ satisfies the same first-order condition as Proposition~\ref{thm:bilateral} with $G$ replacing $F$ and $\xi$ replacing $\theta$.
\end{proposition}

\noindent\textit{Proof.}\;

\medskip\noindent\textit{Step 1: Sufficient statistic.}\; From~\eqref{eq:U_2D}, the counterparty's payoff depends on $(\alpha, \theta)$ only through $\alpha\mu(\theta)$ (the expected contingent payment) and $c(\theta)$ (the cost).  Define $\xi = \alpha\mu(\theta)/c(\theta)$.  Then $U = a + b_1 c(\theta)\xi - c(\theta) - \Phi(K-a; R)$.  For a given contract $(a, b_1)$, two types $(\alpha, \theta)$ and $(\alpha', \theta')$ with the same $\xi$ have the same payoff ranking over contracts, so the IC constraints depend only on $\xi$.

\medskip\noindent\textit{Step 2: Envelope condition.}\; Along the truth-telling locus, the counterparty's rent as a function of $\xi$ satisfies $\frac{dU}{d\xi} = b_1\,c(\theta) > 0$.  Higher-$\xi$ types earn strictly higher rents.  Monotonicity of the allocation in $\xi$ (higher $\xi$ types are more likely to be implemented) is the global IC condition.

\medskip\noindent\textit{Step 3: Virtual surplus.}\; Integrating the rent $\int_{\ubar\xi}^{\xi} b_1 c(\theta(s))\,ds$ by parts over the induced distribution $G(\xi)$ yields the expected information cost:
\[
\E\bigl[\text{Rent}\bigr] \;=\; \int_{\hat\xi}^{\bar\xi} b_1\,\frac{1-G(\xi)}{g_\xi(\xi)}\,g_\xi(\xi)\,d\xi,
\]
where $\hat\xi$ is the cutoff.  Subtracting this from the expected productive surplus and financing cost gives the virtual surplus~\eqref{eq:VS_2D}, which has the same structure as the baseline~\eqref{eq:VS_main}: productive value minus financing cost minus hazard-rate distortion.

\medskip\noindent\textit{Step 4: First-order condition.}\; The principal maximizes $\int_{\hat\xi}^{\bar\xi}\psi_{2D}(\xi)g_\xi(\xi)d\xi - a$ over $(a, b_1)$ subject to the participation constraint.  The FOC for $a$ equates $\Phi_\ell(K-a; R)$ to $\frac{\partial}{\partial a}\E[b_1\frac{1-G(\xi)}{g_\xi(\xi)} \mid q = 1]$, which has the same form as~\eqref{eq:FOC_main}.  \hfill$\square$


\begin{thebibliography}{99}

\bibitem[\protect\citeauthoryear{Acemoglu et al.}{2012}]{acemoglu2012network}
Acemoglu, D., V.\ M.\ Carvalho, A.\ Ozdaglar, and A.\ Tahbaz-Salehi (2012).
The network origins of aggregate fluctuations.
\emph{Econometrica}, 80(5), 1977--2016.

\bibitem[\protect\citeauthoryear{Barrot}{2016}]{barrot2016input}
Barrot, J.-N.\ (2016).
Trade credit and industry dynamics: Evidence from trucking firms.
\emph{Journal of Finance}, 71(5), 1975--2016.

\bibitem[\protect\citeauthoryear{Biais et al.}{2007}]{biais2007optimal}
Biais, B., T.\ Mariotti, G.\ Plantin, and J.-C.\ Rochet (2007).
Dynamic security design: Convergence to continuous time and asset pricing implications.
\emph{Review of Economic Studies}, 74(2), 345--390.


\bibitem[\protect\citeauthoryear{Bolton and Scharfstein}{1990}]{bolton1990theory}
Bolton, P.\ and D.\ S.\ Scharfstein (1990).
A theory of predation based on agency problems in financial contracting.
\emph{American Economic Review}, 80(1), 93--106.

\bibitem[\protect\citeauthoryear{Burkart and Ellingsen}{2004}]{burkart2004situ}
Burkart, M.\ and T.\ Ellingsen (2004).
In-kind finance: A theory of trade credit.
\emph{American Economic Review}, 94(3), 569--590.

\bibitem[\protect\citeauthoryear{Che and Gale}{1998}]{che1998standard}
Che, Y.-K.\ and I.\ Gale (1998).
Standard auctions with financially constrained bidders.
\emph{Review of Economic Studies}, 65(1), 1--21.


\bibitem[\protect\citeauthoryear{Chetty}{2009}]{chetty2009sufficient}
Chetty, R.\ (2009).
Sufficient statistics for welfare analysis: A bridge between structural and reduced-form methods.
\emph{Annual Review of Economics}, 1(1), 451--488.

\bibitem[\protect\citeauthoryear{Cu\~{n}at}{2007}]{cunat2007trade}
Cu\~{n}at, V.\ (2007).
Trade credit: Suppliers as debt collectors and insurance providers.
\emph{Review of Financial Studies}, 20(2), 491--527.

\bibitem[\protect\citeauthoryear{DeMarzo}{2005}]{demarzo2005optimal}
DeMarzo, P.\ M.\ (2005).
The pooling and tranching of securities: A model of informed intermediation.
\emph{Review of Financial Studies}, 18(1), 1--35.

\bibitem[\protect\citeauthoryear{DeMarzo and Duffie}{1999}]{demarzo1999liquidity}
DeMarzo, P.\ M.\ and D.\ Duffie (1999).
A liquidity-based model of security design.
\emph{Econometrica}, 67(1), 65--99.

\bibitem[\protect\citeauthoryear{DeMarzo and Fishman}{2007}]{demarzo2007optimal}
DeMarzo, P.\ M.\ and M.\ J.\ Fishman (2007).
Optimal long-term financial contracting.
\emph{Review of Financial Studies}, 20(6), 2079--2128.

\bibitem[\protect\citeauthoryear{DeMarzo, Kremer, and Skrzypacz}{2005}]{demarzo2005bidding}
DeMarzo, P.\ M., I.\ Kremer, and A.\ Skrzypacz (2005).
Bidding with securities: Auctions and security design.
\emph{American Economic Review}, 95(4), 936--959.


\bibitem[\protect\citeauthoryear{Gertner, Scharfstein, and Stein}{1994}]{gertner1994internal}
Gertner, R.\ H., D.\ S.\ Scharfstein, and J.\ C.\ Stein (1994).
Internal versus external capital markets.
\emph{Quarterly Journal of Economics}, 109(4), 1211--1230.

\bibitem[\protect\citeauthoryear{Giannetti, Burkart, and Ellingsen}{2011}]{giannetti2011what}
Giannetti, M., M.\ Burkart, and T.\ Ellingsen (2011).
What you sell is what you lend? Explaining trade credit contracts.
\emph{Review of Financial Studies}, 24(4), 1261--1298.

\bibitem[\protect\citeauthoryear{Holmstr\"{o}m and Tirole}{1997}]{holmstrom1997financial}
Holmstr\"{o}m, B.\ and J.\ Tirole (1997).
Financial intermediation, loanable funds, and the real sector.
\emph{Quarterly Journal of Economics}, 112(3), 663--691.

\bibitem[\protect\citeauthoryear{Inderst and Laux}{2005}]{inderst2003laux}
Inderst, R.\ and C.\ Laux (2005).
Incentives in internal capital markets: Capital constraints, competition, and investment opportunities.
\emph{RAND Journal of Economics}, 36(1), 215--228.

\bibitem[\protect\citeauthoryear{Innes}{1990}]{innes1990limited}
Innes, R.\ D.\ (1990).
Limited liability and incentive contracting with ex-ante action choices.
\emph{Journal of Economic Theory}, 52(1), 45--67.

\bibitem[\protect\citeauthoryear{Kaplan and Str\"{o}mberg}{2003}]{kaplan2003financial}
Kaplan, S.\ N.\ and P.\ Str\"{o}mberg (2003).
Financial contracting theory meets the real world: An empirical analysis of venture capital contracts.
\emph{Review of Economic Studies}, 70(2), 281--315.

\bibitem[\protect\citeauthoryear{Klapper}{2006}]{klapper2006role}
Klapper, L.\ (2006).
The role of factoring for financing small and medium enterprises.
\emph{Journal of Banking and Finance}, 30(11), 3111--3130.

\bibitem[\protect\citeauthoryear{Klapper, Laeven, and Rajan}{2012}]{klapper2012trade}
Klapper, L., L.\ Laeven, and R.\ Rajan (2012).
Trade credit contracts.
\emph{Review of Financial Studies}, 25(3), 838--867.

\bibitem[\protect\citeauthoryear{Li}{2021}]{li2021mechanism}
Li, Y.\ (2021).
Mechanism design with financially constrained agents and costly verification.
\emph{Theoretical Economics}, 16(3), 1139--1194.

\bibitem[\protect\citeauthoryear{Laffont and Tirole}{1986}]{laffont1986using}
Laffont, J.-J.\ and J.\ Tirole (1986).
Using cost observation to regulate firms.
\emph{Journal of Political Economy}, 94(3), 614--641.

\bibitem[\protect\citeauthoryear{Liu}{2016}]{liu2016optimal}
Liu, T.\ (2016).
Optimal equity auctions with heterogeneous bidders.
\emph{Journal of Economic Theory}, 166, 94--123.

\bibitem[\protect\citeauthoryear{Liu and Bernhardt}{2019}]{liu2019optimal}
Liu, T.\ and D.\ Bernhardt (2019).
Optimal equity auctions with two-dimensional types.
\emph{Journal of Economic Theory}, 184, 104913.

\bibitem[\protect\citeauthoryear{Manelli and Vincent}{2007}]{manelli2007multidimensional}
Manelli, A.\ M.\ and D.\ R.\ Vincent (2007).
Multidimensional mechanism design: Revenue maximization and the multiple-good monopoly.
\emph{Journal of Economic Theory}, 137(1), 153--185.

\bibitem[\protect\citeauthoryear{Milgrom and Segal}{2002}]{milgrom2002envelope}
Milgrom, P.\ and I.\ Segal (2002).
Envelope theorems for arbitrary choice sets.
\emph{Econometrica}, 70(2), 583--601.

\bibitem[\protect\citeauthoryear{Mussa and Rosen}{1978}]{mussa1978monopoly}
Mussa, M.\ and S.\ Rosen (1978).
Monopoly and product quality.
\emph{Journal of Economic Theory}, 18(2), 301--317.

\bibitem[\protect\citeauthoryear{Myers and Majluf}{1984}]{myers1984capital}
Myers, S.\ C.\ and N.\ S.\ Majluf (1984).
Corporate financing and investment decisions when firms have information that investors do not have.
\emph{Journal of Financial Economics}, 13(2), 187--221.

\bibitem[\protect\citeauthoryear{Myerson}{1981}]{myerson1981optimal}
Myerson, R.\ B.\ (1981).
Optimal auction design.
\emph{Mathematics of Operations Research}, 6(1), 58--73.

\bibitem[\protect\citeauthoryear{Nachman and Noe}{1994}]{nachman1994optimal}
Nachman, D.\ C.\ and T.\ H.\ Noe (1994).
Optimal design of securities under asymmetric information.
\emph{Review of Financial Studies}, 7(1), 1--44.

\bibitem[\protect\citeauthoryear{Ng, Smith, and Smith}{1999}]{ng1999evidence}
Ng, C.\ K., J.\ K.\ Smith, and R.\ L.\ Smith (1999).
Evidence on the determinants of credit terms used in interfirm trade.
\emph{Journal of Finance}, 54(3), 1109--1129.

\bibitem[\protect\citeauthoryear{Petersen and Rajan}{1997}]{petersen1997trade}
Petersen, M.\ A.\ and R.\ G.\ Rajan (1997).
Trade credit: Theories and evidence.
\emph{Review of Financial Studies}, 10(3), 661--691.

\bibitem[\protect\citeauthoryear{Scharfstein and Stein}{2000}]{scharfstein1998dark}
Scharfstein, D.\ S.\ and J.\ C.\ Stein (2000).
The dark side of internal capital markets: Divisional rent-seeking and inefficient investment.
\emph{Journal of Finance}, 55(6), 2537--2564.

\bibitem[\protect\citeauthoryear{Stein}{1997}]{stein1997internal}
Stein, J.\ C.\ (1997).
Internal capital markets and the competition for corporate resources.
\emph{Journal of Finance}, 52(1), 111--133.


\bibitem[\protect\citeauthoryear{Stiglitz and Weiss}{1981}]{stiglitz1981credit}
Stiglitz, J.\ E.\ and A.\ Weiss (1981).
Credit rationing in markets with imperfect information.
\emph{American Economic Review}, 71(3), 393--410.

\end{thebibliography}
\end{document}